\documentclass[11pt,twoside]{article}
\usepackage{fancyhdr}

\usepackage[letterpaper, margin=0.9in]{geometry}

\usepackage{amsmath, amssymb, amsthm,thmtools}
\usepackage{mathtools}
\usepackage{enumitem, subcaption, graphicx, hyperref}
\usepackage{paper-defs}
\usepackage{macros}
\usepackage{xcolor}
\usepackage[normalem]{ulem}

\usepackage{booktabs}
\usepackage{rotating}
\usepackage{tabularx}

\setlist[itemize]{
  itemsep=2pt,
  parsep=0pt,
  topsep=4pt,
  partopsep=0pt,
  leftmargin=1.5em
}

\setlist[enumerate]{
  itemsep=2pt,
  parsep=0pt,
  topsep=4pt,
  partopsep=0pt,
  leftmargin=1.5em
}

\usepackage{siunitx}
\usepackage[
  backend=biber,
  style=numeric-comp,
  sorting=none,        % ← number in citation order
  sortcites=true,      % ← sort numbers within a \cite{...}
  defernumbers=true,   % ← assign numbers by order of first citation
  giveninits=true,
  maxcitenames=2,
  maxbibnames=10,
  minbibnames=3,
  doi=false,
  isbn=false,
  url=false,
  eprint=true
]{biblatex}

\AtEveryBibitem{
  \clearfield{month}
  \clearfield{day}
  \clearlist{language}%
  \ifentrytype{article}
  {
    \iffieldundef{journaltitle}{}{
      \clearfield{eprint}
      \clearfield{eprinttype}
      \clearfield{eprintclass}
    }
  }{}
}

\newcommand{\headers}[2]{%
  \fancyhead[CO]{#1}
  \fancyhead[CE]{#2}
}

\headers{Sparse Dynamics via Kernel Collocation}%
        {A. Hsu, I. Griss Salas, J.M. Stevens-Haas, J.N. Kutz, A.Y. Aravkin, B. Hosseini}

\bibliography{references}

\title{A joint optimization approach to identifying sparse dynamics using least squares kernel collocation}

\author{
  Alexander W.~Hsu\thanks{Department of Applied Mathematics, University of Washington, Seattle, WA, USA} \thanks{Corresponding author: \texttt{owlx@uw.edu}.}
  \and Ike~Griss~Salas\footnotemark[1]
  \and Jacob~M.~Stevens-Haas\footnotemark[1]
  \and J.~Nathan~Kutz\footnotemark[1]
  \and Aleksandr~Aravkin\footnotemark[1]
  \and Bamdad~Hosseini\footnotemark[1]
}
\date{}
\begin{document}
\maketitle
\begin{abstract}
We develop an all-at-once modeling framework for learning systems of ordinary differential 
equations (ODE) from scarce, partial, and noisy observations of the states. The proposed methodology 
amounts to a combination of sparse recovery strategies for the ODE over a function library
combined with techniques from reproducing kernel Hilbert space (RKHS) theory for estimating 
the state and discretizing the ODE. Our numerical experiments reveal that the proposed strategy leads to 
significant gains in terms of accuracy, sample efficiency, and robustness to noise, both in terms of 
learning the equation and estimating the unknown states. This work demonstrates
capabilities well beyond existing and widely used algorithms while extending the modeling flexibility of other recent developments in equation discovery.
\end{abstract}

\section{Introduction}

The identification of ordinary differential equations (ODEs) and dynamical systems 
is a fundamental problem in control \cite{gevers_personal_2006,ljung2010perspectives,ljungTheory}, data assimilation \cite{smith_application_1962,Jazwinski1970},
and more recently in scientific machine learning (ML) \cite{Brunton2016,rackauckas2020universal,Raissi2018Hidden}. While algorithms 
such as Sparse Identification of Nonlinear Dynamics (SINDy) and its variants \cite{Kaptanoglu2022} are widely used by practitioners, they often fail in scenarios where observations of the 
state of the system are scarce, indirect, and noisy. 
In such scenarios modifications to 
SINDy-type methods are required to enforce additional constraints on the recovered 
equations to make them consistent with the observational data. Put simply, 
traditional SINDy-type methods work in two steps: (1) the data is used to filter the 
state of the system and estimate the derivatives, and (2) the filtered state is used to learn the underlying dynamics. 
In the regime of scarce, noisy and incomplete data, step 1 is inaccurate, which can propagate to poor results in the subsequent step 2. 

In this paper, we propose an all-at-once approach to filtering and equation learning based on collocation in a reproducing kernel Hilbert space (RKHS) which we term Joint SINDy (JSINDy), and show that the issues above can be mitigated by performing both steps together.
This joins a broader class of dynamics-informed methods that integrate the governing equations directly into the learning objective, either as hard constraints or as least-squares relaxations, which couples the problems of state estimation and model discovery.
Representative examples include physics-informed and sparse-regression frameworks based on neural networks, splines, kernels, finite differences, and adjoint methods \cite{Raissi2019-so,Chen2021,ijcai2021p283,jalalian2025dataefficientkernelmethodslearning,Fung2025-odr,Hokanson2023-sr,rackauckas2020universal}.

\subsection{Problem setup}
Consider a process 
\(\vect{x}(t) = (x_1(t),...,x_d(t)) \in C^p([0,T],\R^d)\), for \(d \ge 1\), 
along with the observables  \((\vect{y}_n)_{n=1}^N\) that are connected via the  
model 
\begin{align}\label{eqn:general_setup}
    \dtp\vect{x}(t) &
    = f(\xdiffs (t))
    ,\quad \vect{y}_n = \vect{V}_n^\top \vect{x}(t_n) + \vect{\epsilon}_n, 
    \\
    \text{where} \quad \xdiffs (t) &:= \xdiffsbig \label{eqn:D-def},
\end{align}
with known measurement matrices \(\vect{V}_n\) (taken to be the identity when the entire state is observed), observation times \(\vect{t}=(t_n)_{n=1}^N\), and measurement noise  \(\vect{\epsilon}_n\) which are assumed to be independent with zero mean and bounded variance. 
Given these observations, our goal is to estimate the nonlinear map \(f:\R^{dp} \mapsto \R^d\) from an instance 
of the $\vect{y}_n$, i.e., to find \(\widehat{f} \approx f\), and simultaneously infer the process $\vect{x}(t)$, i.e., find \(\widehat{\vect{x}}(t) \approx \vect{x}(t)\). Additionally, we ask that our representation of $\widehat{f}$ is \emph{sparse} in some basis or library of functions \(\dict= \left\{\phi_1,...,\phi_Q\right\}\). That is, we assume 
\begin{equation}
    \label{eqn:f_parameterization}
    \widehat{f}(\xdiffs (t))_m = \sum_{q=1}^Q \widehat{\theta}_{m,q} \phi_q\left(\xdiffs (t)\right), \quad \widehat{f}(\xdiffs) = (\widehat{f}(\xdiffs (t))_1,...,\widehat{f}(\xdiffs (t))_d), 
\end{equation}
where only a small number of \(\widehat{\theta}_{m,q}\) are nonzero in order to have an interpretable symbolic form, following 
the seminal work \cite{Brunton2016,wang2011,rudy2016datadriven,schaeffer2017learning}. As we do not know \textit{a priori} which features \(\phi_q\) are active, we parametrize $\widehat f$ using the full feature library 
and promote sparsity in the coefficients \(\widehat{\vect{\theta}} := (\widehat{\theta}_{m,q}), m=1,\cdots,d, q = 1,\cdots, Q\). 

Our setup poses a number of fundamental challenges: 
\begin{enumerate}

    \item Broadly speaking we would like to find $\widehat{f}$ such that 
    \begin{equation*}
         \widehat{f}(\xdiffs (t)) \approx \dtp\vect{x}(t).
    \end{equation*}
    However, the input to $\widehat{f}$ as well as the right hand side 
    of the above display depend on \(\vect{x}(t)\) and its derivatives, both of which are unknown and not observed directly. 

    \item Exacerbating this issue is the fact that the
    measurements $\vect{y}_n$ are sparse, indirect, and noisy, which 
    leads to further challenges in estimating $\vect{x}$ and its derivatives.    
\end{enumerate}

Estimating and dealing with the errors in the inputs to \(\widehat{f}\) and the target \(\dtp\vect{x}(t)\) is the fundamental property of system identification problems, which differentiates it from standard approximation or ML tasks. 
If \emph{all} of these variables were fully observed, i.e., we have direct access to the exact values of \(((D[\vect{x}](t_n))_{n=1}^N\), 
as well as unbiased estimates of \(\dtp\vect{x}(t_n)\), then this would reduce to a standard sparse regression problem. 

We overcome the above challenges by proposing 
 a joint estimation approach where 
 an approximation  \(\widehat{\vect{x}}(t)\) to the true 
 state \(\vect{x}(t)\) is obtained at the same time 
 as the approximation  \(\widehat{f}\) to the nonlinear dynamics.
 The key idea in our approach is to enforce self 
 consistency of the \(\widehat{\vect{x}}\) and \(\widehat{f}\), 
 i.e., \(\widehat{\vect{x}}\) should approximately satisfy the ODE defined by \(\widehat{f}\) while matching the observed data.
 Since the data is limited and the problem 
 is severely ill-posed, we propose an optimization
 approach to solve the problem with appropriate regularization terms for both the \(\vect{\theta}\) 
 parameters and the approximate state 
 \(\widehat{\vect{x}}\), in particular, we use a 
 sparse regularizer for the former and a 
 reproducing kernel Hilbert space (RKHS)
 penalty for the latter.

\subsection{Summary of contributions}
In view of SINDy and other existing works on equation discovery \cite{bayesiansindy,Brunton2016,Fung2025-odr,Chen2021,ijcai2021p283,Hokanson2023-sr}, we summarize our contributions as follows:
\begin{enumerate}
  \item We develop an RKHS regularized least-squares approach to learning general ODEs directly from sparse and partial observations. 
  (see Section~\ref{sec:continuous-system-id}).

  \item We establish a representer theorem that converts the infinite-dimensional problem into a finite-dimensional nonlinear least-squares system,  
  and design an alternating optimization algorithm based on the Levenberg–Marquardt (LM) method and sparsifying iterations for support selection  
  (see Sections~\ref{subsec:discretization} and~\ref{sec:implementation}).

  \item We demonstrate the flexibility and robustness of JSINDy across a broad range of problems, including scarce and partial observations, higher-order systems, and misspecified models,  
  demonstrating accurate recovery in various challenging data regimes  
  (see Section~\ref{sec:numerics}).
\end{enumerate}
Our formulation of the problem differs from previous works in two important ways. We allow for arbitrary order ODEs, while previous works 
are tailored to $p=1$, and we specifically consider partial, scarce, and 
scattered observations while previous works often assume access to high resolution full 
state observations, i.e., $\vect{y}_n = \vect{x}(t_n) + \vect{\epsilon}_n$ on a fine grid.

\subsection{Literature Review}
\label{subsec:lit-review}
Learning continuous time dynamical systems from time series data has been a prominent problem in system identification \cite{aastrom1971system, ljung2010perspectives, keesman2011system}. Seminal works using symbolic regression \cite{bongard2007automated,schmidt2009distilling} through genetic programming sought to discover physical laws from observational data. The works \cite{wang2011,Brunton2016,rudy2016datadriven,schaeffer2017learning} developed and popularized the use of a sparsity prior over a feature library to promote parsimony in the discovered model, drawing on ideas from compressive sensing \cite{donoho2006compressed,candes2006robust,candes2008intro}.
Bayesian approaches have also been considered, which allow for uncertainty quantification and motivate new strategies for model selection \cite{zhang2018robust,niven2019bayesian,yang2020bayesian,bayesiansindy, north2022bayesian, north2023bayesian}. 

Many of these works consider a two-step framework which first estimates time-derivatives, then performs regression to estimate the dynamics. These approaches are straightforward to implement and often succeed when measurements of the process are sufficiently dense but fail when the data is scarcely sampled and noisy. 
These approaches first solve a smoothing problem to estimate the state and its derivatives, using methods such as total variation regularization \cite{Chartrand2011-cv}, kernel regression \cite{LONG2024134095}, Savitzky–Golay filters, smoothing splines \cite{Varah1982-ay}, or Kalman smoothing \cite{haas2024kalman}.
These approaches inevitably trade bias for robustness: as shown in \cite{vanBreugel2020NumericalDifferentiation}, noise-resistant estimators lie on a Pareto frontier balancing error magnitude and error correlation due to implicitly assuming an inaccurate underlying dynamical model.
An alternative to estimating derivatives, but that still lies within a two-step framework is to adopt a weak-form formulation \cite{Messenger2021a}, which avoids explicit differentiation by matching integrals against test functions.
However, this still requires estimating integrals from noisy and scarcely sampled data, a problem that remains challenging despite being more well posed than differentiation.
Consequently, any two-step approach to learning dynamics is inherently limited by the absence of a proper dynamical model in its initial estimation stage.

To address limitations of two-step approaches, dynamics-informed approaches jointly estimate the dynamics and the state, allowing each to reinforce the other.
These include 
Physics-Informed Neural Networks-Sparse Regression (PINN-SR) \cite{Chen2021} and Physics-Informed Spline Learning (PiSL) \cite{ijcai2021p283} which use least squares collocation applied to neural networks and splines respectively, while obtaining sparsity in the coefficients of the ODE by alternating between applying variants of gradient descent and sequentially thresholded ridge regression. 
\cite{jalalian2025dataefficientkernelmethodslearning} takes a least squares collocation approach using RKHS methods, while also representing the unknown differential equation in an RKHS.
The ODR-BINDy method in \cite{Fung2025-odr} applies a least squares relaxation to a finite difference discretization, and handles feature selection by greedily maximizing a Bayesian evidence over models defined by feature sets. \cite{Ribera2022-nz} also applies a least squares relaxation to a finite difference discretization, but optimizes using a variational annealing approach, and focuses on identifying dynamics with hidden variables. Notably, our objective function \eqref{eqn:jsindy-opt} is closely related to that of \cite{Fung2025-odr,Ribera2022-nz,ijcai2021p283,Lahouel2024} and \cite{Chen2021}, though we apply different discretization, variable selection, and optimization strategies. 
\cite{Hokanson2023-sr} applies a finite difference discretization but imposes hard ODE constraints by solving the KKT system associated to the discretized problem. They promote sparsity via \(\ell_p\) regularization, and optimize using iteratively reweighted least squares. 

Aside from \cite{Hokanson2023-sr}, which considers decreasing sampling frequencies, and \cite{Ribera2022-nz}, which considers hidden variables, none of these works deal with very low sampling frequencies, partial observations through general linear measurements, and higher order dynamics in the context of ODEs. These issues are, however, considered heavily in \cite{Chen2021} and \cite{jalalian2025dataefficientkernelmethodslearning} in the context of PDEs, where partial information on the state (being infinite dimensional) and higher order derivatives are fundamental to the problem. \cite{Hokanson2023-sr} crucially point out that decoupling the level of discretization in their finite difference stencil from the sampling rate of the data is required for mitigating bias in estimating the ODE when the sampling rate of the data is low, and demonstrate this in a restricted setting on the Van der Pol system. We generally consider sampling rates over an order of magnitude lower than other works and flexibly handle partial observations and higher order ODEs. 

\subsection{Outline}
We derive JSINDy in \Cref{sec:continuous-system-id}, describe the algorithmic implementation in \Cref{sec:implementation}, give extensive numerical results in \Cref{sec:numerics}, and discuss the ODE relaxation in \Cref{sec:discussion}. 

\subsection{Notation}
We will use boldface for vectors, matrices, and vector-valued trajectories and use subscripts to denote indices, e.g., \(\vect{x}(t) = (x_1(t),x_2(t),...,x_d(t))\). We adopt semicolon notation for parametrized functions, such as \(\vect{x}(t;\vect{z})\), and use a dot \((\cdot)\) in place of an argument in order to treat an object as a function. For a function \(k\) of two arguments, and vectors \(\vect{t},\vect{\tau}\), we will write \(k(\vect{t},\vect{\tau})\) to mean the matrix with entries \(k(t_i,\tau_m)\). We will write \(\partial_{j}^{r}\) to denote \(r\)-th derivative with respect to the \(j\)-th argument. 

\section{Joint state estimation and system identification}
\label{sec:continuous-system-id}
In this section, we formulate our approach as an infinite dimensional optimization problem in the product of an RKHS with \(\mathbb{R}^{d \times Q}\). We discretize this objective via a quadrature rule and develop a representer theorem which reduces the optimization to a finite dimensional subspace, and prove existence of solutions to both the discretized and continuous problems. 

\subsection{The abstract formulation}

Our approach centers around the joint 
computation of the estimated state $\widehat{\vect{x}}(t)$, henceforth called \emph{state estimation}, as 
well as the right hand side $\widehat{f}$, 
henceforth called \emph{system identification}. 
We will require that the pair 
$\widehat{\vect{x}}$ and $\widehat{f}$ are \emph{self 
consistent}, in the sense that $\widehat{\vect{x}}$
should (approximately) satisfy the ODE defined by $\widehat{f}$:
\begin{equation*}
    \dtp\widehat{\vect{x}}(t) = \widehat{f}\left(\widehat{\vect{x}}(t), \frac{d}{dt} \widehat{\vect{x}}(t), \dots,\frac{d^{p-1}}{dt^{p-1}}\widehat{\vect{x}}(t)\right).
\end{equation*}
To this end, we impose self-consistency using a least squares penalty method \cite{tristan_analysis,tristan_pde,sasha_tristan_partialmin,vanLeeuwen2013Mitigating,xu2025efficientkernelbasedsolversnonlinear,Raissi2019-so}, which is discretized via collocation in an RKHS. 
Imposing that \(\widehat{f}\) has a sparse representation in the library \(\dict = \{\phi_1,\ldots,\phi_Q\}\), we adopt a parametrized notation for how it operates on $\vect x$ and its derivatives:
\begin{equation}
    f(\xdiffs (t);\vect{\theta}) := 
    \left(
        \sum_{q=1}^{Q}\theta_{1,q}\phi_q\left(\xdiffs (t)\right),
        \ldots,
        \sum_{q=1}^{Q}\theta_{d,q}\phi_q\left(\xdiffs (t)\right)
    \right), \quad \vect{\theta} \in \R^{d\times Q}.
\end{equation}
For instance, if we consider the first order constant coefficient linear ODE \(\frac{d}{dt} \vect{x}(t) = \vect{A} \vect{x}\),
then \(\xdiffs(t) = \vect{x}\), and we would be looking for \(\vect{\theta} = \vect{A}\). 

We define our estimators \((\widehat{\vect{x}},\widehat{\vect{\theta}})\) as solutions to the following optimization problem. 
\begin{equation}\label{eqn:jsindy-opt-sparse}
\begin{aligned}
\operatorname*{minimize}_{\vect{x}\in\X,\vect{\theta}\in\R^{d\times Q}}\quad 
    & \LL(\vect{x},\vect{\theta}) \\
\text{subject to}\quad 
    & \supp(\vect{\theta}) \subset \mathcal{S}\left(\vect{x}\right),\\
    \LL(\vect{x}, \vect{\theta}) 
        & := \alpha\, \LL_D(\vect{x}) 
        + \beta\, \LL_C(\vect{x},\vect{\theta})
        + \frac{\lambda}{2}\|\vect{x}\|_\X^2 
        + \frac{\mu}{2}\|\vect{\theta}\|_F^2, \\
    \LL_D(\vect{x}) 
        & := \frac{1}{2} \sum_{n=1}^N 
            \left|\vect{V}_n^\top\vect{x}(t_n) - \vect{y}_n\right|^2, \\
    \LL_C(\vect{x},\vect{\theta}) 
        & := \frac{1}{2}\int_{0}^{T}
            \left|\dtp\vect{x}(t) - f\left(\xdiffs(t);\vect{\theta}\right)\right|^2 dt.
\end{aligned}
\end{equation}

where
\begin{itemize}
    \item \(\alpha,\beta,\lambda,\mu >0\) are relative weights trading off data fidelity, self consistency through the ODE, and regularization of the trajectory estimate and coefficients. 
    \item \(\LL_D\) is a data misfit term
enforcing that the estimated trajectory \(\widehat{\vect{x}}\) is consistent with the observed data. 
    \item \(\LL_C\) is a least squares penalty enforcing that \(\widehat{\vect{x}}\) approximately satisfies the ODE 
induced by \(\widehat{f}\).
    \item \(\xdiffs (t)\), defined in \cref{eqn:D-def} expands the trajectory \(\vect{x}(t)\) into the states and derivatives.
    \item \(\X\) is a vector-valued RKHS used to represent the trajectories, whose elements are smooth enough such that \(\xdiffs (t)\) is well defined pointwise for all \(t\in [0,T]\), and \(\|\vect{x}\|_{\mathcal{X}}\) promotes regularity in the estimates of the trajectory. See the discussion below and \Cref{subsec:sobolev} for more technical details.
    
    \item \(\mathcal{S}\) is a feature selection procedure for choosing active terms from the feature library \(\Phi\) given a fixed trajectory estimate, and maps to a subset of indices. This can be seen as applying a traditional SINDy procedure, with an arbitrary choice of sparsifier, and using exact derivatives of the approximation of the trajectory in order to choose the subset of features which are active.
    
    Because we handle sparsity through the feature selection operator \(\mathcal{S}\) we do not need to apply regularization to \(\vect{\theta}\) to promote sparsity. In principle, however, the Frobenius norm \(\frac{\mu}{2}\|\vect{\theta}\|_F^2\) could be replaced by a more general sparsity promoting regularizer \(R(\vect{\theta})\) which would supplant \(\mathcal{S}\).
    \end{itemize}

The inclusion of a feature selection procedure is motivated by a natural algorithmic idea for sparse nonlinear optimization. 
Consider the fixed point iteration 
\begin{align}\label{eq:iteration-abstract}
\begin{aligned}
(\vect{x}^{(k+1)},\vect{\theta}^{(k+1)})
&\in
\left\{
\begin{aligned}
&\argmin_{\vect{x}\in\X,\vect{\theta}\in\R^{d\times Q}} 
    && \LL(\vect{x},\vect{\theta}) \\
&\text{subject to}\quad 
    && \supp(\vect{\theta}) \subset S^{(k)}
\end{aligned}
\right. \\
S^{(k+1)} &= \mathcal{S}\left(\vect{x}^{(k+1)}\right),
\end{aligned}
\end{align}
which alternates between optimizing on an active set of features, and choosing features to use in the active set. Given a solution \((\vect{x}^*,\vect{\theta}^*)\) to \eqref{eqn:jsindy-opt-sparse}, define \(S^* = \mathcal{S}(\vect{x}^*)\). If 
\begin{equation}
(\vect{x}^*,\vect{\theta}^*)\in
\operatorname*{arg\,min}_{\substack{\vect{x}\in\X,\ \vect{\theta}\in\R^{d\times Q}\\ \supp(\vect{\theta})\subset S^\star}}
\mathcal{L}(\vect{x},\vect{\theta}),
\end{equation}
then the pair \((\vect{x}^*,\vect{\theta}^*)\) is a fixed point of (\ref{eq:iteration-abstract}). This is a strong condition, but is likely to hold whenever \(\mathcal{S}\) is a reasonably good sparsifier that is consistent with the collocation loss \(\LL_C\). This is also a desirable property in itself; we want for our estimate fixed to the active set to be the best estimate possible on that active set.

The ADO algorithms of \cite{ijcai2021p283,Chen2021} are exactly of the form (\ref{eq:iteration-abstract}) and motivate it as a surrogate to \(\ell_0\) minimization. Similarly, the greedy evidence maximization in \cite{Fung2025-odr} is related in that it solves a sequence of fixed-support nonlinear least squares problems, though in each iteration, separate subproblems are solved for each candidate feature to remove. This has the benefit of minimizing a fixed objective function, and ensuring convergence due to the monotone feature removal process. This is discussed in detail algorithmically in \Cref{subsec:optimization} and \Cref{subsec:sparsifiers}.

In our theoretical analysis, we will remove the constraint and work with
\begin{equation}\label{eqn:jsindy-opt}
    \operatorname*{minimize}_{\vect{x}\in\X, \vect{\theta}\in\R^{d\times Q}}\quad 
    \LL(\vect{x},\vect{\theta}),
\end{equation}
which describes the first piece of the fixed point iteration and still entails solving a challenging nonlinear programming problem. 

This approach combines regularization on the trajectory estimate with a least squares relaxation of the differential equation, which prevents trajectory estimates from diverging away from the observed data. Further, this relaxation allows the observed data itself to guide the optimization in the early iterations, pulling our state estimates towards a neighborhood of the true trajectory even before our estimates effectively approximate the solution to the ODE. 

To represent the \(d\)-dimensional trajectory \(\vect{x}:[0,T]\to\mathbb{R}^d\), we work in a vector-valued RKHS \(\mathcal{X}\) of functions \([0,T] \mapsto \R^d\) generated by the \(d\)-fold Cartesian product of an RKHS \(\mathcal{H}_k\), associated to a symmetric positive definite kernel \(k\). This is equivalent to modeling each coordinate \(x_i(t)\) of the trajectory \(\vect{x}(t)\) separately in the RKHS \(\mathcal{H}_k\), or by taking \(\X\) to be the vector-valued RKHS with matrix valued kernel \(K(t,t') =  k(t,t')I_d\).
We will require that for derivative orders \(j = 0,...,p\), and times \(t \in [0,T]\), the derivative evaluation map \(\vect{x} \mapsto \frac{d^j}{dt^j}\vect{x}(t)\) is a bounded linear operator on \(\mathcal{X}\) (see \Cref{subsec:sobolev}).
It suffices to take \(k\) such that \(\mathcal{H}_k\) is compactly embedded in \(C^p([0,T])\), which is satisfied if \(k \in C^{p,p}([0,T]\times[0,T])\) \cite{ZhouDerivative}.

We choose the kernel $k$ to be the sum of a high-order Matérn kernel \cite{Rasmussen2006Gaussian} and a constant kernel to account for systems where the time-average is nonzero. More precisely, with $h_\nu$ denoting the  $\nu$-order Matérn covariance function and $\ell$ as a length scale parameter, we use kernels of the form
\begin{equation}\label{eqn:kernel}
    k(t,t') = c_0 + c_1 h_{\nu}\left(\frac{|t - t'|}{\ell}\right),\quad 
    h_{\nu}(r) = \frac{2^{1-\nu}}{\Gamma(\nu)}
\left(\sqrt{2\nu}r\right)^{\nu}
K_{\nu}\left(\sqrt{2\nu}\,r\right)
\end{equation}
where \(K_{\nu}\) is the modified Bessel function of the second kind, and \(c_0,c_1>0\). These kernels induce RKHSs which satisfy the boundedness condition of the derivative evaluation maps whenever \(\nu > p + \frac{1}{2}\). See \Cref{subsec:sobolev} for 
more details on kernels and smoothness.

As a preliminary step in our model, we choose the kernel hyperparameters $c_0,c_1,\ell$, and the variance of the observation noise $\sigma^2 = \mathbb{V}\left[\epsilon\right]$ using maximum 
likelihood and treat the distribution of \(\epsilon\) as Gaussian, as is common in 
Gaussian process regression; see for example \cite{Rasmussen2006Gaussian}. 
From there, \(\|\vect{x}\|_\X^2\) corresponds to the regularization induced by 
the RKHS norm associated to this fitted Gaussian process.

\subsection{Discretization}\label{subsec:discretization}

In order to obtain a computable objective function, we discretize the ODE residual $\LL_C(\vect{x},\vect{\theta})$ using a quadrature rule, setting
\begin{align}\label{eq:quadrature}\LL_C(\vect{x},\vect{\theta}) 
    &= \frac{1}{2}\int_{0}^{T}\left|\dtp\vect{x}(\tau) - f\left(\xdiffs (\tau);\vect{\theta}\right)\right|^2d\tau \\
    &\approx \frac{1}{2}\sum_{m=1}^{M} w_m \left|\dtp \vect{x}(\tau_m) - f\left(\xdiffs(\tau_m);\vect{\theta}\right)\right|^2 := \LL_C^{\vect{\tau}}(\vect{x},\vect{\theta})
\end{align}
where we refer to $\vect{\tau} = (\tau_m)_{m=1}^M$ as collocation points, and $\vect{w} = (w_m)_{m=1}^M$ as collocation weights. In all of our examples, we take uniformly spaced collocation nodes and uniform weights; this is discussed further in \Cref{sec:discussion}. Further, for notational convenience and without loss of generality, we assume that the observation times \(\vect{t} = (t_n)_{n=1}^N\) are a subset of the collocation points \(\vect{\tau}\), which can be enforced by augmenting \(\vect{\tau}\), and taking some \(w_m\) to be zero. 

Substituting \(\LL_C^{\vect{\tau}}\) for \(\LL_C\) in \cref{eqn:jsindy-opt}, 
we obtain the optimization problem
\begin{equation}\label{eqn:obj-semidiscrete}
\begin{aligned}
\underset{\vect{x}\in\X,\, \vect{\theta}\in \R^{d\times Q}}{\text{minimize}}
\LL^{\vect{\tau}}(\vect{x},\vect{\theta})&:=
\frac{\alpha}{2} \sum_{n=1}^N \left| \vect{V}_n^\top\vect{x}(t_n) - \vect{y}_n\right|^2 
\\
& + \frac{\beta}{2}\sum_{m=1}^{M} w_m
\left|\dtp \vect{x}(\tau_m) - f\left(\xdiffs(\tau_m);\vect{\theta}\right)\right|^2
\\
& + \frac{\lambda}{2}\|\vect{x}\|_\X^2 + \frac{\mu}{2}\|\vect{\theta}\|_F^2
\end{aligned}
\end{equation}
which we use to compute our estimators. At this point, the reproducing property and boundedness of derivative evaluations in the RKHS \(\X\) give a natural parametrization of \(\widehat{\vect x}\) which is justified by the following \emph{representer theorem}, which converts optimization over the infinite-dimensional space \(\X\) into a problem on a finite-dimensional subspace. This result is more general than the typical representer theorems often found in ML, such as in \cite{Rasmussen2006Gaussian}, which only make use of pointwise evaluations, but is very similar to the representer theorem given in \cite{CHEN2021110668}, and follows from the more general results of \cite{Owhadi2019-xh}. 
The purpose of this result is to reduce the optimization problem over the infinite-dimensional space \(\X\) into an optimization problem over the finite-dimensional subspace spanned by an explicit set of basis functions. 

\begin{theorem}\label{thm:obj-representer}
    Consider the optimization problem (\ref{eqn:obj-semidiscrete}) for \(\alpha,\beta,\lambda,\mu>0\). 
    % , and let the \(\frac{\mu}{2}\|\vect{\theta}\|_F^2\) be coercive and lower semi-continuous. 
    Let \(\X\) be the vector-valued RKHS associated to the \(d\)-fold Cartesian product of an RKHS \(\mathcal{H}_k\) associated to the positive definite kernel \(k\). Assume that linear operators of the form \(\vect{x} \mapsto \frac{d^r}{dt^r}\vect{x}(\tau)\) for \(r=0,...,p\) and \(\tau\in [0,T]\) are bounded in the RKHS \(\X\). Also, assume that all functions \(\phi_q\) in the library \(\dict\) are continuous on \(\mathbb{R}^{dp}\).

\begin{enumerate}
    \item There exists a minimizer \((\widehat{\vect{x}},\widehat{\vect{\theta}})\) to the optimization problem
        \begin{equation*}
        \underset{\vect{x}\in\X,\,\vect{\theta}\in\R^{d \times Q}}{\mathrm{minimize}}\LL^{\vect{\tau}}(\vect{x},\vect{\theta}). 
        \end{equation*}

    \item Define the basis functions
\begin{align}\label{eqn:dtau-def}
\kbasis(t)
&:= \big(k(t,\tau_1),\partial_{2}k(t,\tau_1),\dots,\partial_{2}^{p}k(t,\tau_1),\dots,
k(t,\tau_M),\dots,\partial_{2}^{p}k(t,\tau_M)\big)\in\R^{(p+1)M},
\end{align}
and for \(\vect{z}\in\R^{(p+1)M\times d}\), set
\begin{equation}\label{eqn:zmatvec}
\vect{x}(t;\vect{z}) := \kbasis(t)^{\!\top}\vect{z}.
\end{equation}

Then problem \eqref{eqn:obj-semidiscrete} is equivalent to
\begin{equation}\label{eqn:z-obj}
\underset{\vect{z}\in\R^{(p+1)M\times d},\,\vect{\theta}\in\R^{d\times Q}}{\mathrm{minimize}}
\LL^{\vect{\tau}}\!\big(\vect{x}(\cdot;\vect{z}),\vect{\theta}\big),
\end{equation}
 in the sense that they have the same optimal value and \((\widehat{\vect{x}},\widehat{\vect{\theta}})\) is a minimizer if and only if \(\widehat{\vect{x}}\) admits a representation \(\widehat{\vect{x}}(t)=\kbasis(t)^{\!\top}\widehat{\vect{z}}\) for some \(\widehat{\vect{z}}\), where \((\widehat{\vect{z}},\widehat{\vect{\theta}})\) is a solution to \eqref{eqn:z-obj}.

\item For all \(\vect{z}\in\R^{(p+1)M\times d}\),
\begin{equation}\label{eqn:rkhs-norm}
\|\vect{x}(\cdot;\vect{z})\|_{\X}^{2}=\tr\!\big(\vect{z}^{\!\top}\vect{K}\,\vect{z}\big), \quad 
\vect{K} =
\begin{bmatrix}
k(\vect{\tau},\vect{\tau}) & \partial_{2}^{1}k(\vect{\tau},\vect{\tau}) & \cdots & \partial_{2}^{p}k(\vect{\tau},\vect{\tau})\\
\partial_{1}^{1}k(\vect{\tau},\vect{\tau}) & \partial_{1}^{1}\partial_{2}^{1}k(\vect{\tau},\vect{\tau}) & \cdots & \partial_{1}^{1}\partial_{2}^{p}k(\vect{\tau},\vect{\tau})\\
\vdots & \vdots & \ddots & \vdots\\
\partial_{1}^{p}k(\vect{\tau},\vect{\tau}) & \partial_{1}^{p}\partial_{2}^{1}k(\vect{\tau},\vect{\tau}) & \cdots & \partial_{1}^{p}\partial_{2}^{p}k(\vect{\tau},\vect{\tau}). 
\end{bmatrix} \in \R^{(p+1)M\times (p+1)M}.
\end{equation}

\end{enumerate}

\end{theorem}
\begin{proof}
We have that
\begin{equation}
    \LL^{\vect{\tau}}(\vect{x},\vect{\theta}) = \alpha \LL_D(\vect{x}) + \beta \LL_C^{\vect{\tau}}(\vect{x},\vect{\theta}) + \frac{\lambda}{2}\|\vect{x}\|_\X^2 + \frac{\mu}{2}\|\vect{\theta}\|_F^2.
\end{equation}
By nesting, problem \eqref{eqn:obj-semidiscrete} is equivalent to 
\begin{equation}
\underset{\vect{\theta}\in\R^{d\times Q}}{\mathrm{minimize}}\inf _{\vect{x}\in\X}\LL^{\vect{\tau}}(\vect{x},\vect{\theta})
\end{equation}
and we can thus focus solely on the dependence with respect to \(\vect{x}\), studying the inner variational problem, as the outer problem is already finite dimensional. 
Going term by term, we have 
\begin{equation}
    \LL_D(\vect{x}) = \frac{1}{2} \sum_{n=1}^N \left| \vect{V}_n^\top \vect{x}(t_n) - \vect{y}_n\right|^2,
\end{equation}
which depends only on point evaluations of the form \(\vect{x}(\tau_m)\) (recall that we assume the observation times \(\vect{t}\) are a subset of \(\vect{\tau}\)). 
Next, we have that
\begin{equation}
    \LL_C^{\vect{\tau}}(\vect{x},\vect{\theta}) = \frac{1}{2}\sum_{m=1}^{M} w_m \left|\dtp \vect{x}(\tau_m) - f\left(\xdiffs(\tau_m);\vect{\theta}\right)\right|^2,
\end{equation}
depends on \(\vect{x}\) only through point evaluations of \(\vect{x}\) and its derivatives at \((\tau_m)_{m=1}^M\), i.e., the values of \(\xdiffs(\tau_m)\) and \(\frac{d^p}{dt^{p}}\vect{x}(\tau_m)\). 

Recall that linear operators of the form \(\vect{x} \mapsto \frac{d^r}{dt^r}\vect{x}(\tau)\) for \(r=0,...,p\) and \(\tau\in [0,T]\) are assumed to be bounded on the RKHS \(\X\). Thus, by the reproducing property \cite{ZhouDerivative}, 
\begin{equation}
    \frac{d^r}{dt^r}x_i(\tau_m) = \left\langle \partial_2^rk(\cdot ,\tau_m),x_i(\cdot)\right\rangle_{\mathcal{H}_k}. 
\end{equation}
Since the objective \(\LL^{\vect{\tau}}\) depends only on the values of \(\frac{d^r}{dt^r}\vect{x}(\tau_m)\) for \(r = 0,1,...,p\), \(m = 1,...,M\), with the addition of the squared RKHS norm, this implies that 
\begin{equation}
    \label{eqn:trajectory-component}
    \widehat{x}_i \in \operatorname{span}\Big\{ \partial_{2}^{r} k(\,\cdot\,,\tau_m) : m=1,\dots,M, r=0,\dots,p \Big\},
\end{equation}
by \Cref{thm:rep}. 
This is precisely the subspace spanned by the basis \(\kbasis(\cdot)\) in \eqref{eqn:dtau-def}. 
By the reproducing property, we also have that 
\begin{equation}
    \left\langle \partial_2^r k(\cdot ,\tau_m),\partial_2^s k(\cdot ,\tau_{m'})\right\rangle_{\mathcal{H}_k} = \partial_1^r \partial_2^sk(\tau_m,\tau_{m'}),
\end{equation}
which gives the structure of the matrix \(\vect{K}\), 
and the RKHS norm of an individual component of the trajectory \(x_i(\cdot) \in \Span \left\{\partial_2^rk(\cdot ,\tau_m),\,\,r = 0,1,...,p, m = 1,...,M
    \right\}\). 
Stacking the individual components together and noting that 
\begin{equation}
    \vect{x}(t) = (x_1(t),...,x_d(t)) \text{ and } \|\vect{x}\|_\X^2 = \sum_{i=1}^d \|x_i\|_{\mathcal{H}_k}^2
\end{equation} 
produces the compact representation in \eqref{eqn:zmatvec} and the formula for the RKHS norm in \eqref{eqn:rkhs-norm}. 

We obtain existence of minimizers from continuity and coercivity. The objective is continuous with respect to \(\vect{z}\) and \(\vect{\theta}\) due to the continuity of the feature maps \(\phi_q\).  If \(\vect{K}\) is invertible, then the objective is coercive and thus attains a minimum. If it is not invertible, then we may, without affecting \(\LL_D\) or \(\LL_{C}^{\vect{\tau}}\), restrict \(\vect{z}\) to the range of \(\vect{K}\); this is essentially applying a finite-dimensional representer theorem to the subspace spanned by \(\kbasis\). Under this additional constraint, the objective becomes coercive, giving existence of minima. Finally, the equivalence of the optimization problems allows us to transfer existence from problem \ref{eqn:z-obj} to problem \ref{eqn:obj-semidiscrete}. 
\end{proof}

A direct application of \Cref{thm:obj-representer} gives the finite-dimensional objective function 
\begin{equation}\label{eqn:final-obj}
    \begin{aligned}
\Ldiscrete(\vect{z},\vect{\theta})
:=\LL^{\vect{\tau}}(\vect{x}(\cdot;\vect{z}),\vect{\theta})=
\frac{\alpha}{2}& \sum_{n=1}^N \left| \vect{V}_n^\top \vect{x}(t_n;\vect{z}) - \vect{y}_n\right|^2 
\\
 + \frac{\beta}{2}&\sum_{m=1}^{M} w_m
\left|\dtp \vect{x}(\tau_m;\vect{z}) - f\left(\xdiffs(\tau_m;\vect{z})\right)\right|^2
\\
 + \frac{\lambda}{2}&\|\vect{z}\|_{\vect{K}}^2 + \frac{\mu}{2}\|\vect{\theta}\|_F^2 ,\,  \|\vect{z}\|_{\vect{K}}^2 := \zkz
\end{aligned}
\end{equation}

We also give a theorem on the existence of solutions to the continuous problem \cref{eqn:jsindy-opt}.

\begin{theorem}\label{thm:cont-existence}
    Consider the optimization problem (\ref{eqn:jsindy-opt}). 
    Let \(\alpha,\beta,\lambda,\mu>0\).
    Let \(\X\) be the vector-valued RKHS associated to the \(d\)-fold Cartesian product of an RKHS \(\mathcal{H}_k\) associated to the positive definite kernel \(k\). Assume that \(\mathcal{H}_k\) is compactly embedded in \(C^p([0,T],\mathbb{R}^d)\). Assume that all functions \(\phi_q\) in the library \(\dict\) are continuous on \(\mathbb{R}^{dp}\). Then there exists a minimizer \((\vect{x}^*,\vect{\theta}^*) \in \mathcal{X} \times \mathbb{R}^{d\times Q}\).
\end{theorem}
\begin{proof}
    Because the features \(\phi_q\) are assumed to be continuous functions, and point evaluations of derivatives up to order \(p\) are bounded on \(\mathcal{H}_k\) due to the compact embedding property, \(\LL_C\) and \(\LL_D\) are continuous on both \(\mathcal{X} \times \mathbb{R}^{d\times Q}\) and \(C^p([0,T],\mathbb{R}^d) \times \mathbb{R}^{d\times Q}\). Additionally, both are lower bounded by zero. 
    % Because \(R\) is lower semicontinuous on \(\mathbb{R}^{d\times Q}\), we have \(\LL\) is lower semicontinuous on \(C^p([0,T],\mathbb{R}^d)\). 
    Since \(\frac{\lambda}{2}\|\vect{x} \|_{\mathcal{X}}^2 + \frac{\mu}{2}\|\vect{\theta}\|_F^2\) is coercive on \(\mathcal{X} \times \mathbb{R}^{d\times Q}\), we obtain that \(\LL\) is coercive as well. 
    For a minimizing sequence \((\vect{x}^{(k)},\vect{\theta}^{(k)})\), we may extract a subsequence that strongly converges in \(C^p([0,T],\mathbb{R}^d) \times \mathbb{R}^{d\times Q}\) to a candidate minimizer \((\vect{x}^*,\vect{\theta}^*)\). Lower semicontinuity of \(\LL\) yields that \(\lim_{k\to\infty}\LL(\vect{x}^{(k)},\vect{\theta}^{(k)}) = \LL(\vect{x}^*,\vect{\theta}^*)\) giving existence of a minimizer. 
\end{proof}

\section{Algorithms and implementation details}
\label{sec:implementation}

Reintroducing the sparse selection procedure, we obtain the optimization problem 
\begin{equation}\label{eqn:final-opt}
    \begin{aligned}
\operatorname*{minimize}_{\vect{z}\in\R^{(p+1)M\times d}, \vect{\theta}\in\R^{d\times Q}}\quad 
    & \Ldiscrete(\vect{z},\vect{\theta}) \\
\text{subject to}\quad 
    & \supp(\vect{\theta}) \subset \mathcal{S}\left(\vect{x}\right).
\end{aligned}
\end{equation}

In order to numerically compute solutions to \cref{eqn:final-opt}, we alternate between applying a Levenberg-Marquardt (LM) algorithm to solve smooth nonlinear least squares problems with the support of \(\vect{\theta}\) restricted, and a sparse regression algorithm to choose the support of \(\vect{\theta}\). We first write \(\Ldiscrete\) as a regularized nonlinear least squares objective. 
Define the residual functions
\begin{equation}
    \vect{F}_D(\vect{z}) := 
    \begin{bmatrix}
        \vect{V}_{1}^\top \vect{x}(t_{1};\vect{z})-\vect{y}_{1}\\
        \vdots\\
        \vect{V}_{N}^\top\vect{x}(t_{N};\vect{z})-\vect{y}_{N}
    \end{bmatrix}, \quad 
    \vect{F}_{C}(\vect{z},\vect{\theta})
    =
    \begin{bmatrix}
    \sqrt{w_{1}}\big(\,\dtp \vect{x}(\tau_{1};\vect{z}) - f\big(\xdiffs(\tau_{1};\vect{z});\,\vect{\theta}\big)\big)\\[2pt]
    \vdots\\[2pt]
    \sqrt{w_{M}}\big(\,\dtp \vect{x}(\tau_{M};\vect{z}) - f\big(\xdiffs(\tau_{M};\vect{z});\,\vect{\theta}\big)\big).
    \end{bmatrix}
    \end{equation}
With 
\begin{equation}
    \vect{F}(\vect{z},\vect{\theta}) = \begin{bmatrix}\sqrt{\alpha}\vect{F}_{D}(\vect{z})\\
        \sqrt{\beta}\vect{F}_{C}(\vect{z},\vect{\theta})
        \end{bmatrix},
\end{equation}
observe that
\begin{equation}\label{eqn:nls_obj}
    \Ldiscrete(\vect{z},\vect{\theta}) = 
    \frac{1}{2} \left\|\vect{F}(\vect{z},\vect{\theta})\right\|^2
    + \frac{\lambda}{2} \tr(\vect{z}^\top \vect{K}\vect{z}) + \frac{\mu}{2}\|\vect{\theta}\|_F^2.
\end{equation}
The addition of the quadratic regularization term \(\frac{1}{2} \vect{z}^T \vect{K} \vect{z}\) is completely benign in the nonlinear least squares setting, and may be treated explicitly in Gauss-Newton Hessian approximations. We further perform a change of variables using the Cholesky factors of \(\vect{K}\) in order to transform the quadratic regularization into a multiple of the identity, see \Cref{subsec:cholesky}. 
Additionally, note that \(\vect{F}(\vect{z},\vect{\theta})\) is affine with respect to \(\vect{\theta}\), which makes it natural to apply methods for sparse linear least squares for the feature selection procedure \(\mathcal{S}\); see \Cref{subsec:sparsifiers}. 

\subsection{Optimization}\label{subsec:optimization}
We compute minimizers of \eqref{eqn:z-obj} by alternating between \emph{active-set iterations}, and \emph{sparsifying iterations}.

\begin{enumerate}
    \item \textbf{Active-set iterations} refine the estimates of $\vect{z},\vect{\theta}$ under the condition that $\supp(\vect{\theta})\subset S^{(k)}$ for a fixed active set $S^{(k)}$:
    \begin{equation}\label{eq:active-set-step}
        (\vect{z}^{(k+1)},\vect{\theta}^{(k+1)})
        =
        \argmin_{\vect{z},\vect{\theta}} 
            \Ldiscrete(\vect{z},\vect{\theta})  \text{, subject to } \,
             \supp(\vect{\theta}) \subset S^{(k)}
    \end{equation}
    These iterations drop the features not selected in step (2) and solve smooth nonlinear least-squares problems restricted to the identified support using LM. Taking previous parameter values as a warm start greatly improves efficiency. More algorithmic details are discussed in \Cref{subsec:smoothNLS}.

    \item \textbf{Sparsifying iterations} freeze $\vect{z}=\vect{z}^{(k)}$ and compute
    % \begin{align}\label{eq:sparsify-step}
    %     \tilde{\vect{\theta}}^{(k+1)}
    %     &=
    %     \arg \min_{\vect{\theta}}
    %     \frac{1}{2}\,\big\|\vect{F}(\vect{z}^{(k)},\vect{\theta})\big\|^2 + \frac{\mu}{2}\|\vect{\theta}\|_F^2\\
    %     S^{(k+1)} &= \supp\!\big(\tilde{\vect{\theta}}^{(k+1)}\big)
    % \end{align}
    \begin{equation}\label{eq:sparsify-step}
        S^{(k+1)} =  \mathcal{S}(\vect{x}(\cdot;\vect{z}^{(k)}))
    \end{equation}
    in order to estimate a new active set. Once the estimate \(\vect{x}(\cdot,\vect{z})\) of the trajectory is fixed, a natural strategy is to target the sparse regression problem of finding sparse \(\vect{\theta}\) which approximately minimizes the collocation error \(\LL_C^\tau(\vect{x}(\cdot;\vect{z}^{(k)}),\vect{\theta})\). This objective then has a \emph{linear} least squares structure, and \(S^{(k+1)}\) is taken to be the support of the identified solution. 
    This approach can be seen as applying a traditional 2-step SINDy approach to the current trajectory estimate $ \vect{x}(\tau_m;\vect{z}^{(k)})$, using the \emph{exact} derivatives of the \emph{estimated} trajectory.
\end{enumerate}

We initialize with a dense support by setting \(S^{(0)}\) to be all available indices of \(\vect{\theta}\), and alternate these two steps until the support $S^{(k)}$ stabilizes, i.e. \(S^{(k)} = S^{(k-1)}\), which corresponds to a fixed point.
In most of our examples, we take \(\mathcal{S}\) to be the result of sequentially thresholded ridge regression \cite{Brunton2016}, which is also the choice taken in \cite{ijcai2021p283,Chen2021} and makes our algorithm template match their ADO approach. Options for sparsifiers are discussed further in \Cref{subsec:sparsifiers}.

\subsection{Smooth nonlinear least squares}\label{subsec:smoothNLS}
To solve the smooth subproblems arising in the active-set refinement \cref{eq:active-set-step}, we use an LM method, which is a standard and well studied algorithm for nonlinear least squares \cite{levenberg,marquardt,Aravkin2024-bz,wright1999numerical}. For notational convenience, stack the variables \(\vect{z},\vect{\theta}\) and restrict to the current support:
\begin{equation}
    \vect{\eta}=
\begin{bmatrix}
\mathrm{vec}(\vect{z})\\
\mathrm{vec}(\vect{\theta}_{S^{(k)}})
\end{bmatrix},
\qquad
\vect{G}(\vect{\eta})=\vect{F}(\vect{z},\vect{\theta}).
\end{equation}
The subproblem in \cref{eq:active-set-step} takes the form
\begin{equation}
    \min_{\vect{\eta}} \Ldiscrete(\vect{\eta}):=\frac{1}{2}\big\|\vect{G}(\vect{\eta})\big\|^2 +\frac{1}{2}\vect{\eta}^\top \vect{Q}\vect{\eta},
\end{equation}
where we define the positive semi-definite matrix
\[
\vect{Q}
:=
\begin{bmatrix}
\lambda\,(I_d \otimes \vect{K}) & 0\\[2pt]
0 & \mu\, I_{|S^{(k)}|}
\end{bmatrix}, \; \text{ so that }
\frac{1}{2}\,\vect{\eta}^\top \vect{Q}\,\vect{\eta}
= \frac{\lambda}{2}\,\tr(\vect{z}^\top \vect{K}\vect{z}) + \frac{\mu}{2}\,\|\vect{\theta}\|_F^2.
\]
Note that we have dropped rows/columns for inactive coefficients, i.e., we only keep \(\vect{\theta}_{S^{(k)}}\).

Let \(\vect{M}\) be a positive definite matrix defining a step metric. In our case, we set \(\vect{M} = \vect{Q}\). LM linearizes inside of the least-squares term and adds a quadratic damping \(\gamma^{(k)}\) to control step size and force convergence. The iteration is given by
\begin{align}
\vect{\eta}^{(k+1)}&=\arg\min_{\vect{\eta}}\; m^{(k)}(\vect{\eta})+\frac{\gamma^{(k)}}{2}(\vect{\eta}-\vect{\eta}^{(k)})^\top \vect{M}(\vect{\eta}-\vect{\eta}^{(k)}),\\
m^{(k)}(\vect{\eta})&:=\frac{1}{2}\left\|\nabla \vect{G}(\vect{\eta}^{(k)})(\vect{\eta}-\vect{\eta}^{(k)})+\vect{G}(\vect{\eta}^{(k)})\right\|^2+\frac{1}{2}\vect{\eta}^\top \vect{Q}\vect{\eta}.
\end{align}

We update \(\gamma^{(k)}\) adaptively. Define the \emph{gain ratio} (observed vs. predicted decrease)
\begin{equation}
    \rho^{(k)} = \frac{\Ldiscrete(\vect{\eta}^{(k)}) - \Ldiscrete(\vect{\eta}^{(k+1)})}{\Ldiscrete(\vect{\eta}^{(k)}) - m^{(k)}(\vect{\eta}^{(k+1)})}.
\end{equation}
If \(\rho^{(k)}< 0.01\), we reject the step and increase \(\gamma^{(k)}\); otherwise we accept and update \(\gamma^{(k)}\) for the next iteration. One practical choice, analyzed in \cite{Aravkin2024-bz}, is
\begin{equation}
    \gamma^{(k+1)}=\begin{cases}
        \gamma^{(k)}\,c & \rho^{(k)}<0.4,\\[2pt]
        \gamma^{(k)}     & \rho^{(k)}\in[0.4,0.9],\\[2pt]
        \gamma^{(k)}/c & \rho^{(k)}>0.9,
        \end{cases}
\end{equation}
for some \(c>1\); we use \(c=1.2\).

\subsection{Sparsifying iterations}
\label{subsec:sparsifiers}
As mentioned in \Cref{subsec:optimization}, we alternate between solving the fixed support, smooth nonlinear least squares problem (\ref{eq:active-set-step}) and updating the support of the coefficient set to obtain \(S^{k+1}\), applying a sparse feature selection method in \cref{eq:sparsify-step}. 
Hence, a crucial component of our algorithm is the choice of regularizer or feature selection procedure in \cref{eq:sparsify-step}. The work that introduced SINDy \cite{Brunton2016} suggested ridge regularized sequentially thresholded least squares (STLSQ). This alternates between solving a ridge regularized least squares problem, and evicting features whose coefficients fall below a given threshold. 
Another approach is to add regularization which promotes sparsity. 
    These include LASSO \cite{tibshirani1996lasso} which regularizes with an \(\ell_1\) penalty, SR3\cite{SR3, Champion2020} which minimizes smoothed versions of nonsmooth penalties,
    and MIOSR \cite{bertsimas_learning_2023}, which directly solves \(\ell_0\) constrained problems using mixed integer programming. 
Ensembling the approaches above via bagging and feature subsampling can sometimes improve the performance of such algorithms \cite{fasel_ensemble-sindy_2022}, and provide some uncertainty quantification \cite{ensemble_uncertainty}. 

An alternative approach to feature selection is to define an appropriate sparsity promoting prior distribution and likelihood, and approximately sample from the posterior. These include spike-and-slab priors \cite{Gao2022} and the regularized horseshoe prior \cite{bayesiansindy}.

In most of our experiments, we apply ridge regularized STLSQ as it empirically performed the best, and was easy to tune. In  \cref{subsubsec:lorenz-bench} we applied SR3 \cite{SR3} with an L1 penalty--we found this to be more consistent on varying trajectory lengths and noise levels without any retuning, but it can be difficult to intuit good parameters a priori. In certain cases, including ensembling within our sparsification procedure greatly improved the robustness of our method, see \cref{subsubsec:nonlin-osc}. At the cost of losing a straightforward variational problem, considering an algorithmic sparsification framework allows the specific choice of sparsifier to become an implementation detail that is easy for users to modify and provides substantial flexibility. 

\section{Numerical results}\label{sec:numerics}
We validate JSINDy on a range of examples that illustrate its capability in handling noisy, scarce and partial observations, higher order systems, and model misspecification.\footnote{An implementation of JSINDy and the code used to generate the experiments in this section are available at \url{https://github.com/AHsu98/jsindy}.}

For each example, we generate a true trajectory \(\vect{x}(t)\) solving the ODE \(\dtp \vect{x}(t) = f(\xdiffs,\vect{\theta})\), and take noisy measurements \(\vect{y}_n = \vect{V_n}^\top \vect{x}(t_n) + \vect{\epsilon}_n\) by adding Gaussian noise \(\vect{\epsilon}_n\).
Fitting our model results in state estimates \(\widehat{\vect{x}}(t)\), and system coefficients \(\widehat{\vect{\theta}}\).
We evaluate the state estimates, \(\widehat{\vect x}(t)\) against the true trajectory \(\vect{x}(t)\) using root mean square error (RMSE) averaged across coordinates, and normalized by the variance across time of each coordinate.

\begin{equation}
    \mse^2 := \frac{1}{T\cdot d} 
    \sum_{m=1}^d \frac{\int_0^T (\widehat{x}_m(t) - x_m(t))^2 dt}{\int_0^T (\overline{x_m} - x_m(t))^2 dt},\quad \overline{x_m} := \frac{1}{T}\int_0^T x_m(s)ds
\end{equation}

Where appropriate, we evaluate the fitted coefficient matrix \(\widehat {\vect{\theta}}\) using the normalized Frobenius norm compared to the true coefficients \(\vect{\theta}\).

\begin{equation}
    \mae := \frac{\|\widehat{\vect{\theta}} - \vect{\theta}\|_F}{\|\vect{\theta}\|_F}
\end{equation}

This is not relevant in \Cref{subsec:part-obs} due to unidentifiability, and in \Cref{subsec:miss-spec}, as there is no ``true" coefficient matrix. All models, for each experiment, use a sum of a scaled constant kernel and Mat\'ern kernel with \(\nu = 11/2\), which corresponds to five times differentiable functions. 

We test our algorithm on a variety of systems and observation patterns. We start with the Lorenz 63 system, which is commonly used for testing methods for identifying nonlinear dynamics, using a relatively standard sampling frequency and noise level. Next, we consider the regime of scarce sampling, with sampling frequencies far below what have been previously considered in the literature, working with the Lorenz and Lotka--Volterra dynamics. We then consider problems with partial observations, where we never observe the entire state at once (even with noise), with cases of alternating between observing the different coordinates, and observing only a subset. Finally, we consider the related problem of identifying higher order ODEs from only observations of the state variable, but not its derivatives, working with the Van der Pol system, a higher dimensional coupled version of the Duffing system, and the nonlinear pendulum under model misspecification. The problem of identifying higher order ODEs from point observations can be seen as a version of a partially observed problem with additional structure between dimensions if reduced to first order. 

In all of these examples, we find that the assumption that \(\vect{x}(t)\) satisfies \emph{some} autonomous ODE is a strong prior that greatly improves inference, even when the ODE must itself be identified. Identification of the ODE itself is further facilitated by the assumption of sparsity in the dynamics which significantly restricts the class of models under consideration. The benefits of a joint approach to learning as opposed to two-step estimation have been observed in many works \cite{Fung2025-odr,Chen2021,ijcai2021p283,Hokanson2023-sr,Kaheman_2022,jalalian2025dataefficientkernelmethodslearning}. The problem of data scarcity was considered in \cite{Chen2021,jalalian2025dataefficientkernelmethodslearning,Hokanson2023-sr}. The SIDDS algorithm of \cite{Hokanson2023-sr} found that such an approach can reach the Cramer-Rao lower bound and demonstrated the necessity of decoupling the discretization of the ODE from the grid of observations in order to reduce bias when the sampling period is large. Our work expands the flexibility to even scarcer data and partial observations. 

We largely use the same hyperparameters throughout. We set the trajectory regularization \(\lambda = 10^{-3}\), the data weight \(\alpha = 1\), the collocation penalty \(\beta = 10^{5}\), and use \(500\) uniformly spaced collocation points with uniform weight. Except for \Cref{subsec:comparisons} where we apply SR3, we use STLSQ as a sparsifier, with thresholds and ridge regularization parameters chosen for each problem. See \Cref{table:numerics} for full details of hyperparameters and problem parameters.

Additionally, we include two systematic experiments: one varying the time horizon and noise level on the Lorenz system, and another varying the noise level on a nonlinear damped oscillator system. These show promising empirical results compared to other competitive models. 

\setlength{\aboverulesep}{0pt}
\setlength{\belowrulesep}{0pt}
\setlength{\tabcolsep}{2pt}
\renewcommand{\arraystretch}{0.95}

\begin{table}[ht]
\centering
\resizebox{\linewidth}{!}{%

% \begin{tabular}{|l|S|S|c|S|l|S|l|S|S|
%                 S[table-format=1.2e-1]|
%                 S[table-format=1.2e-1]|}
% \toprule \textbf{Experiment}\rule{0pt}{2.6ex} & \textbf{Noise abs} & \textbf{Noise rel (\%)} & \textbf{Order} &
% \boldmath{$\Delta t$} & \textbf{Obs.} & \boldmath{$t_1$} & \textbf{Obs. Pattern} &
% \textbf{Thr} & \textbf{Ridge} & \(\mse\) & \(\mae\) \\
% \midrule
% Scarce Lorenz        & 0.0 &  0.00 & 1 & 0.5   & 21 (\(\times\)3)  & 10  & Full              & 0.5  & 0.01 & 1.61e-3 & 2.36e-3 \\
% \hline
% Lotka–Volterra       & 0.2 & 19.68 & 1 & 3     & 17 (\(\times\)2)  & 50  & Full              & 0.1  & 0.1  & 1.13e-2 & 3.32e-2 \\
% \hline
% Lorenz Streams       & 0.32 &  6.92 & 1 & 0.025 & 404               & 10  & One at a time     & 0.25 & 0.01 & 5.69e-3 & 1.96e-2 \\
% \hline
% Lorenz Unobserved    & 0.32 &  3.98 & 1 & 0.025 & 404               & 10  & 3rd coord. unobs. & 0.1  & 0.01 & 6.84e-3& {N/A}   \\
% \hline
% Van der Pol          & 0.1 &  7.39 & 2 & 2     & 50                & 100 & Pos.-only         & 0.1  & 0.01 & 4.54e-2& 1.86e-1 \\
% \hline
% Coupled Duffing      & 0.2 & 16.85 & 2 & 0.2   & 150 (\(\times\)5) & 30  & Pos.-only         & 0.1  & 0.01 & 3.55e-2 & 3.63e-2 \\
% \hline
% Nonlin pend. (linear)   & 0.1 &  4.17 & 2 & 1     & 100               & 100 & Pos.-only         & 0.1  & 0.01 & 1.25e-01 & {N/A}   \\
% \hline
% Nonlin pend. (cubic) & 0.1 &  4.17 & 2 & 1     & 100               & 100 & Pos.-only         & 0.1  & 0.01 & 2.22e-2 & {N/A}   \\
% \bottomrule
% \end{tabular}

\begin{tabular}{|l|S|S|c|S|S|l|l|S|S|
                S[table-format=1.2e-1]|
                S[table-format=1.2e-1]|}
\toprule \textbf{Experiment}\rule{0pt}{2.6ex} & \textbf{Noise abs} & \textbf{Noise rel (\%)} & \textbf{Order} &
\boldmath{$\Delta t$} & \boldmath{$t_1$} & \textbf{Obs.} & \textbf{Obs. Pattern} &
\textbf{Thr} & \textbf{Ridge} & \(\mse\) & \(\mae\) \\
\midrule
Lorenz               & 2.0 & 24.60 & 1 & 0.05 & 10 & 200 (\(\times\)3) & Full 
                     & 0.5 & 0.01 & 2.35e-2 & 3.97e-2 \\
\hline
Scarce Lorenz        & 0.0 & 0.00 & 1 & 0.5   & 10  & 21 (\(\times\)3)  & Full              & 0.5  & 0.01 & 1.61e-3 & 2.36e-3 \\
\hline
Lotka–Volterra       & 0.2 & 19.68 & 1 & 3     & 50  & 17 (\(\times\)2)  & Full              & 0.1  & 0.1  & 1.13e-2 & 3.32e-2 \\
\hline
Lorenz Streams       & 0.32 & 6.92 & 1 & 0.025 & 10  & 404               & One at a time     & 0.25 & 0.01 & 5.69e-3 & 1.96e-2 \\
\hline
Lorenz Unobserved    & 0.32 & 3.98 & 1 & 0.025 & 10  & 404               & 3rd coord. unobs. & 0.1  & 0.01 & 6.84e-3 & {N/A}   \\
\hline
Van der Pol          & 0.1 & 7.39 & 2 & 2     & 100 & 50                & Pos.-only         & 0.1  & 0.01 & 4.54e-2 & 1.86e-1 \\
\hline
Coupled Duffing      & 0.2 &16.85 & 2 & 0.2   & 30  & 150 (\(\times\)5) & Pos.-only         & 0.1  & 0.01 & 3.55e-2 & 3.63e-2 \\
\hline
Nonlin. pend. (linear)   & 0.1 & 4.17 & 2 & 1     & 100 & 100               & Pos.-only         & 0.1  & 0.01 & 1.25e-01 & {N/A}   \\
\hline
Nonlin. pend. (cubic) & 0.1 & 4.17 & 2 & 1     & 100 & 100               & Pos.-only         & 0.1  & 0.01 & 2.22e-2 & {N/A}   \\
\hline
Damped nonlin. osc. & 0.32 & 38.22 & 1 & 0.05 & 10 & 200 (\(\times 2\)) & Full & 0.08 & 100.00 & 6.60e-2 & 2.47e-1\\
\bottomrule
\end{tabular}

}
\caption{Summary of experiments. The \(n (\times d)\) indicates that the observations were made at \(n\) time points of the \(d\) dimensional state. \textbf{Thr.} and \textbf{Ridge} are the threshold and ridge regularization strength used in the STLSQ sparsifier.}
\label{table:numerics}
\end{table}

\noindent
\textbf{Lorenz System.} As a first example which is commonly used to test methods for identifying nonlinear dynamics, we considered the Lorenz system: 
\begin{equation}\label{eqn:true_lorenz}
\begin{aligned}
    \frac{dx_1}{dt} &= -10 x_1 + 10 x_2, \\
    \frac{dx_2}{dt} &= 28 x_1 - x_2 - x_1 x_3, \\
    \frac{dx_3}{dt} &= -\frac{8}{3} x_3 + x_1 x_2. \\
\end{aligned}
\end{equation}
We use the initial condition \(\vect{x}(0) = (-8,8,27)\) and took full observations with a sampling period \(\Delta t = 0.05\), and Gaussian noise with variance \(\sigma^2 = 4\).

\begin{figure}
    \centering
    \includegraphics[width=1\linewidth]{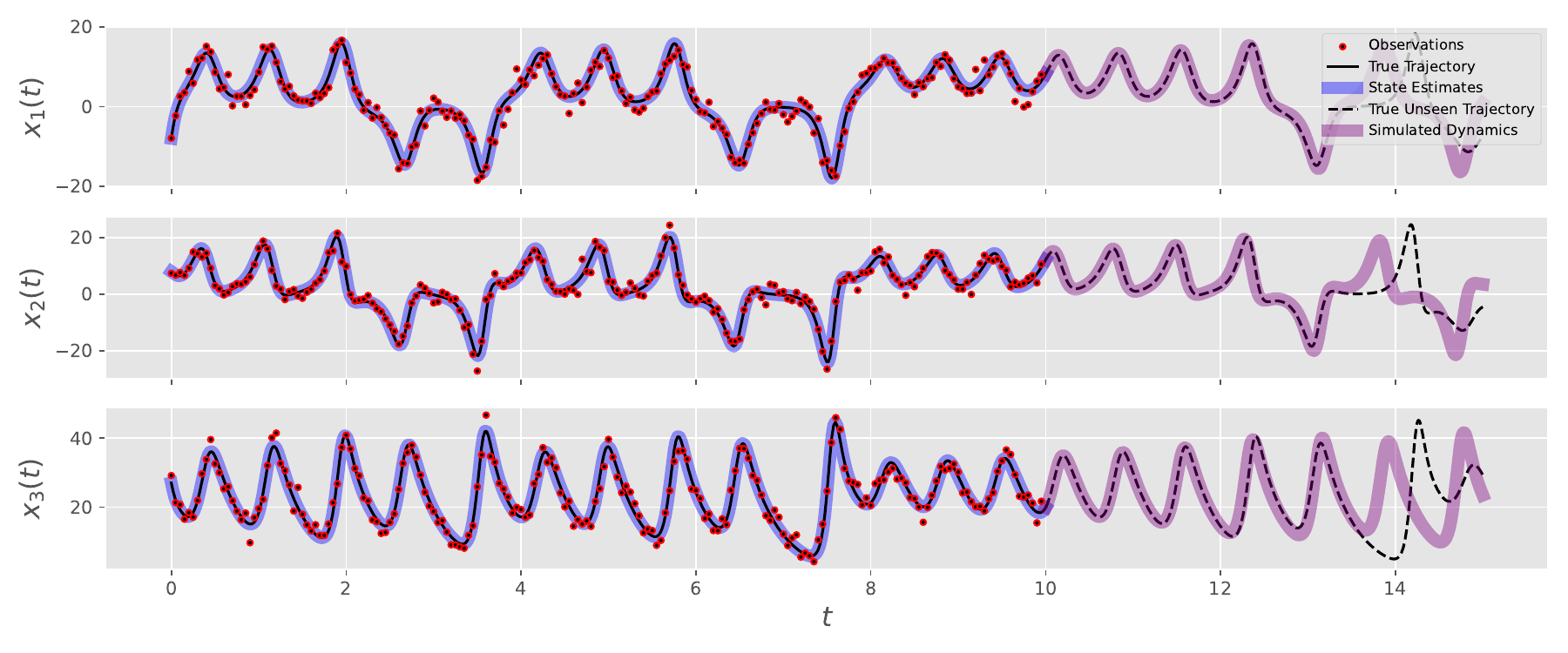}
    \caption{Results for Lorenz 63 system. Observations sampled at $\Delta t=0.05$ until  \(t=10\) with added Gaussian noise of variance \(\sigma^2=4\). Simulated dynamics from \cref{eqn:lorenz} are compared against true dynamics from \(t=10\) until \(t=15\).
    }
    \label{fig:lorenz}
\end{figure}

Our approach obtained the equations
\begin{equation}
\begin{aligned}
    \frac{d x_1}{dt} &= -9.607 x_1 + 9.689 x_2 \\
    \frac{d x_2}{dt} &= 29.122 x_1 - 1.239 x_2 - 1.033 x_1 x_3 \\
    \frac{d x_3}{dt} &= -2.624 x_3 + 0.999 x_1 x_2. \\
    \label{eqn:lorenz}
\end{aligned}
\end{equation}
with \(\mse = 2.353\cdot 10^{-2}\) and \(\mae = 3.975 \cdot 10^{-2} \). Results can be seen in \Cref{fig:lorenz}. This is essentially a baseline example. For this problem, the sampling rate was at the lower end of the standard range that is considered in the literature. Our approach both identified the correct active features, and approximated the coefficient values to reasonable accuracy. In the rest of our numerical examples, we largely consider examples where traditional two-step methods are not applicable, or will obviously break down.

\subsection{Scarce observations}
We now consider cases of very scarcely sampled data with the Lorenz 63 and the Lotka--Volterra systems. 

\subsubsection{Lorenz 63 Scarce} 
We modified the previous example, instead sampling 21 observations of the entire state over the interval \([0,10]\) with a sampling period of \(\Delta t=0.5\), and no noise.  Results can be seen in \Cref{fig:lorenz_scarce} with the following learned system below:

\begin{equation}
\label{eqn:scarce-lorenz}
    \begin{aligned}
    \frac{dx_1}{dt} &= -10.008\,x_1 + 10.006\,x_2 \\
    \frac{dx_2}{dt} &= 28.073\,x_1 - 1.012\,x_2 - 1.003\,x_1 x_3 \\
    \frac{dx_3}{dt}&= -2.663\,x_3 + 1.001\,x_1 x_2
    \end{aligned}
\end{equation}
with \(\mse = 1.608\cdot 10^{-3}\) and \(\mae = 2.356\cdot 10^{-3}\). These results are very accurate, owing to the lack of noise. Indeed, a parameter counting argument suggests that once the support is identified, this problem is essentially a noiseless and well-posed though ill-conditioned inversion problem, where we should expect good algorithms to attain very high accuracy. 

\begin{figure}
    \centering
    \includegraphics[width=1\linewidth]{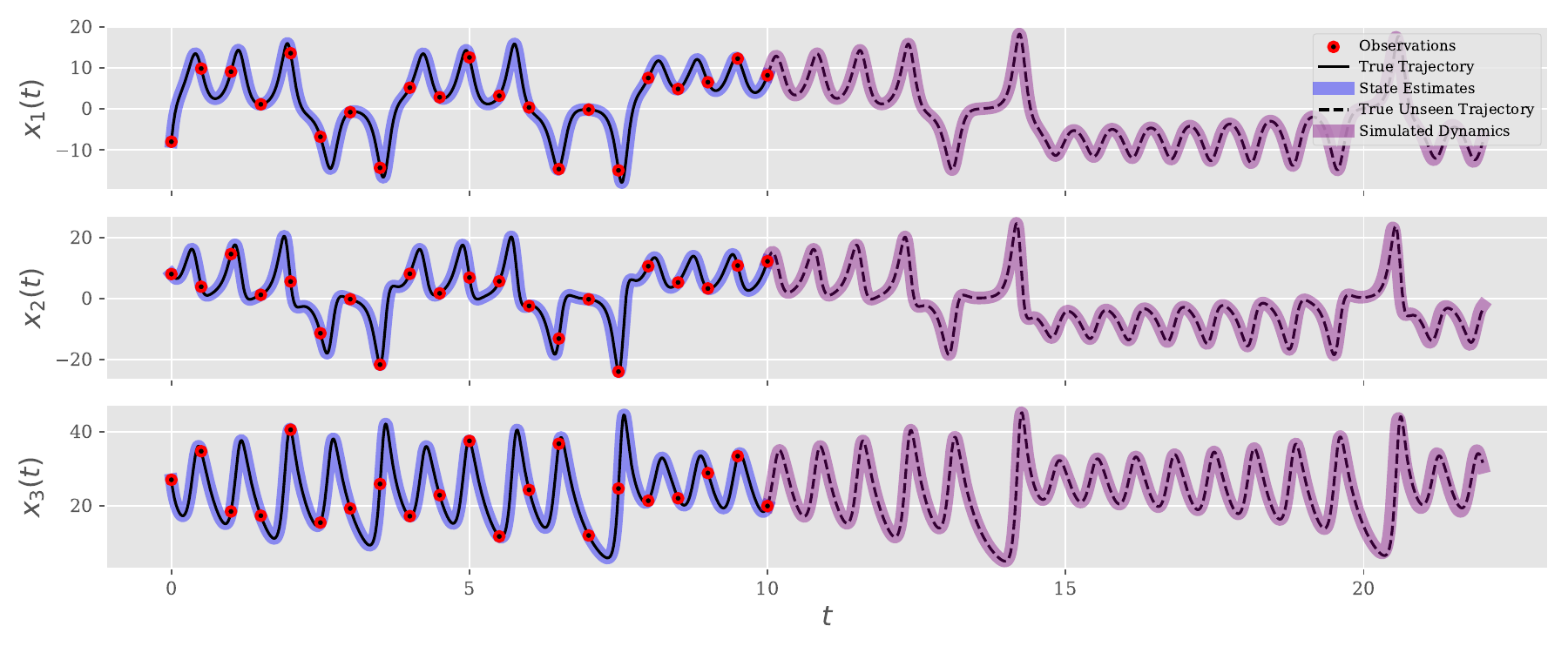}
    \caption{Lorenz 63 system with observations sampled at $\Delta t=0.5$ until  \(t=10\) with no added noise. Simulated dynamics of \eqref{eqn:scarce-lorenz} are plotted against true, unseen, trajectory from \(t=10\) to \(t=25\).
    }
    \label{fig:lorenz_scarce}
\end{figure}

\subsubsection{Lotka--Volterra}\label{subsubsec:lv}
The Lotka--Volterra system is a nonlinear oscillator with periodic solutions, describing population dynamics of a predator-prey ecology, and is a common benchmark for system identification \cite{bertsimas_learning_2023,fasel_ensemble-sindy_2022,Kaheman_2022,haas2024kalman}.  
This system governs the populations of prey $x_1$ and predators $x_2$:
\begin{equation}
    \begin{split}
        \frac{dx_1}{dt} &=  x_1 -  x_1 x_2,\\
        \frac{dx_2}{dt} &= - x_2 +  x_1x_2,
    \end{split}
    \label{eqn:true-lv}
\end{equation}
 We again considered scarce sampling, with approximately \(2\) observations per cycle, taken at a sampling period of \(\Delta t=3\) in the interval \(t\in [0,50]\).  Gaussian noise with variance $\sigma^2=0.04$ was added to the observations.
We set the feature library to include monomials up to second order:
\begin{equation*}
    \Phi = \left\{1, x_1, x_2, x_1^2, x_1x_2, x_2^2\right\}. 
\end{equation*}

Our approach obtained the system
\begin{equation}
    \begin{split}
        \frac{dx_1}{dt} &= 1.009 x_1 - 1.056 x_1 x_2 \\[6pt]
        \frac{dx_2}{dt}&= - 0.980 x_2 + 0.500 x_1 x_2.
    \end{split}
    \label{eqn:learned-lv}
\end{equation}
with \(\mse = 1.131\cdot 10^{-1}\) and \(\mae =3.323\cdot 10^{-2}\). 
We also compare the phase portrait of the learned model to the true dynamics, simulating multiple trajectories from various initial conditions in \Cref{fig:lv-phase-portrait}. 

\begin{figure}
    \centering
    \includegraphics[width=1\linewidth]{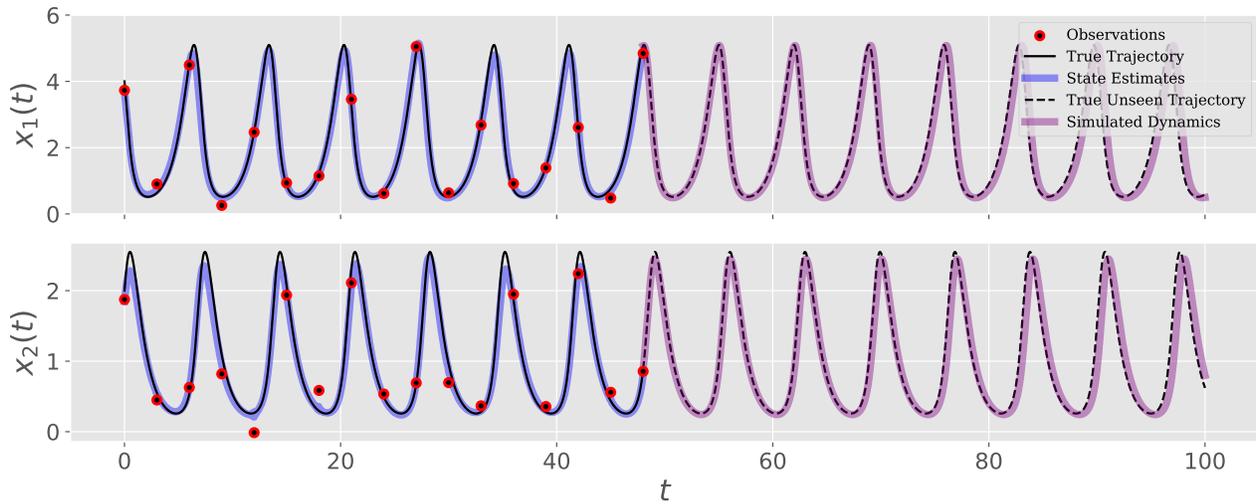}
    \caption{Results for the Lotka--Volterra system of \cref{subsubsec:lv}. Samples are taken at the rate \(\Delta t=3\) up to \(t=48\), including additive noise with \(\sigma = 0.2\).  Simulated dynamics are shown up to \(t=100\), comparing the results from \eqref{eqn:learned-lv} with the true dynamics.}
    \label{fig:lv-results}
\end{figure}
\begin{figure}
    \centering
    \includegraphics[width=
    \linewidth]{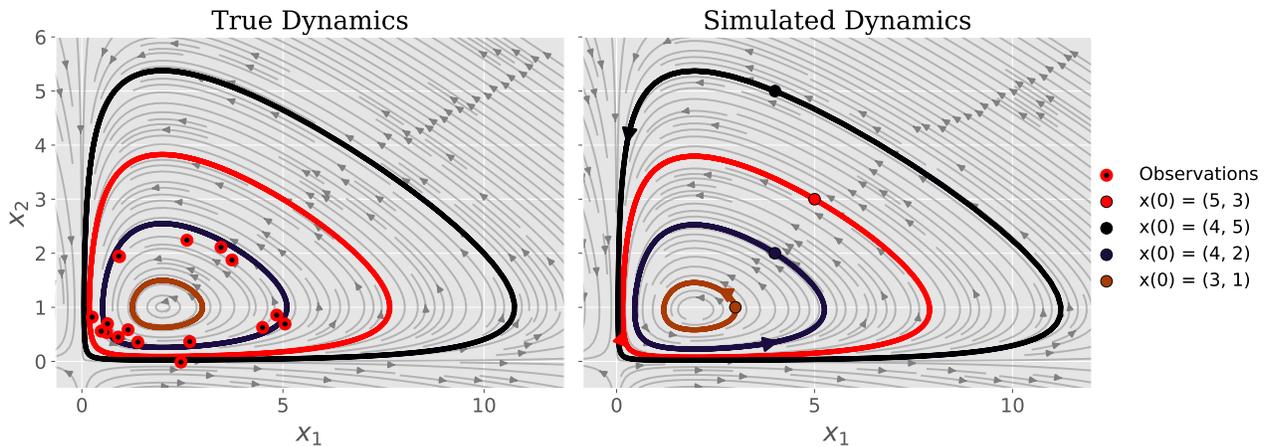}
    \caption{Left: True Lotka--Volterra plotted from multiple initial conditions on phase portrait using \eqref{eqn:true-lv}. Red trajectory and measurements are those from  \Cref{fig:lv-results}. Right: Using learned model \eqref{eqn:learned-lv} to simulate same trajectories provided same initial conditions. }
    
    \label{fig:lv-phase-portrait}
\end{figure}

\subsection{Partial observations}\label{subsec:part-obs}
We now consider problems with only partial observations. Rather than observing the entire state \(\vect{x}(t_n)\), we observe a sequence of linear measurements of the state, \(\vect{y}_n = \vect{V}_n^\top\vect{x}(t_n)\). 
We test our algorithm on the Lorenz dynamics with two challenging observation patterns.
In these examples, we took each \(\vect{V}_n\) to be a standard basis vector, observing one coordinate at a time. 
 In both of these examples, we took the same initial conditions as above, set the sampling period \(\Delta t = 0.025\), and set the noise \(\sigma^2 = 0.1\). 

\subsubsection{Alternating coordinates}
In the first case, we alternated between periods of observing each of the three coordinates one at a time for \(10\) time steps each. For each coordinate, this leads to relatively long gaps (of duration \(t=0.5\) ) where no observations of the state are made. 
However, our approach allows the dynamics to effectively fill in this missing information, correlating the data available between the different coordinates. Notably, if the dynamics were \textit{known ahead of time}, then applying a high order nonlinear Kalman smoothing method could be similarly effective, but this is not applicable here, since there is no a priori knowledge of the governing equations \cite{Bell1994IteratedKalmanSmoother,Sarkka2013BayesianFilteringSmoothing}. Our model learned the system in \eqref{eqn:learn-lorenz-streams} with \(\mae = 5.692\cdot 10^{-3}\) and \(\mse = 1.961\cdot 10^{-2}\). Results can be seen in  \Cref{fig:lorenz_streams}.

\begin{equation}
\begin{aligned}
\frac{d x_1}{d t} &= -9.987\, x_1 + 9.951\, x_2,\\
\frac{d x_2}{d t} &= 28.104\, x_1 - 1.030\, x_2 - 1.003\, x_1 x_3,\\
\frac{d x_3}{d t} &= -0.606\, \cdot 1 - 2.642\, x_3 + 1.005\, x_1 x_2.
\end{aligned}
\label{eqn:learn-lorenz-streams}
\end{equation}

\begin{figure}
    \centering
    \includegraphics[width=1.\linewidth]{  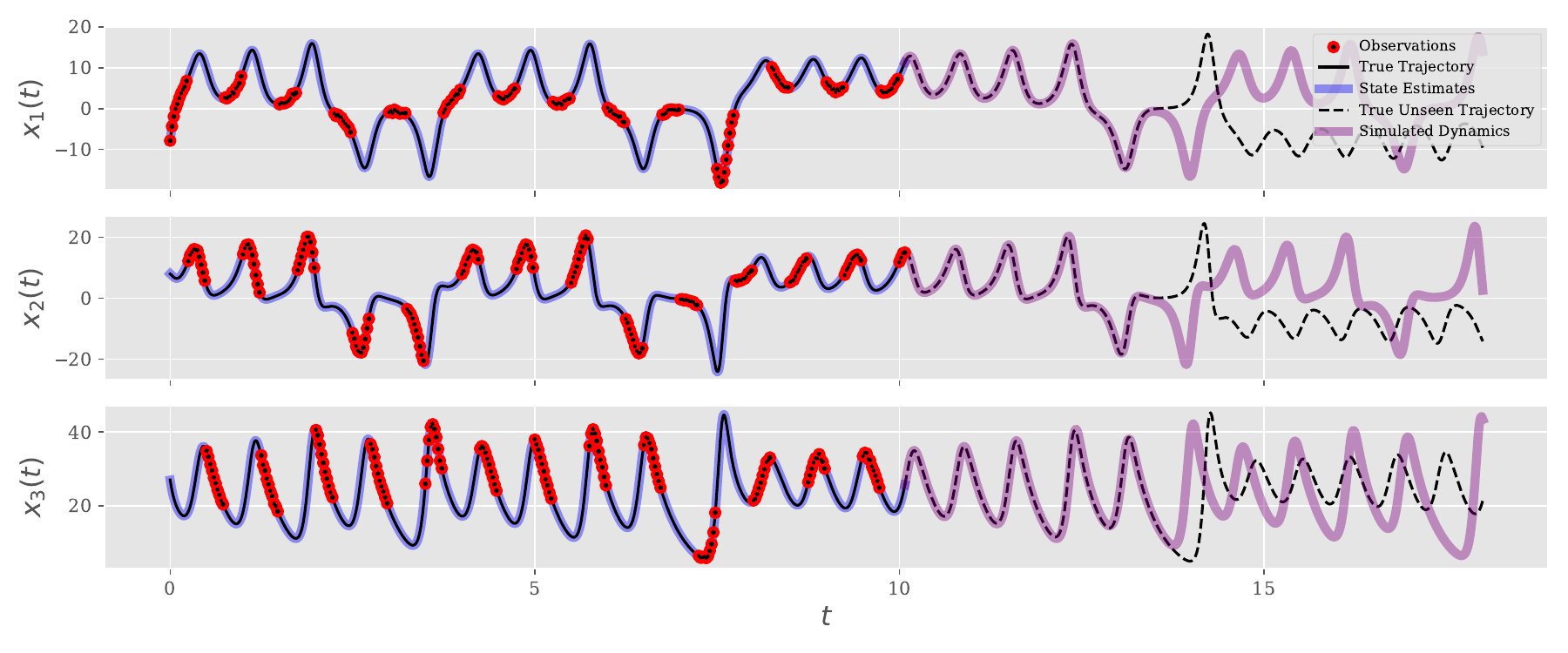}
    \caption{ Results for Lorenz 63 system with streams of partial observations. Observations alternate between each coordinate every 10 samples at sampling rate \(\Delta t = 0.025\) up to \(t=10\) with added noise set to \(\sigma^2 = 0.1\). Learned dynamics are simulated from \(t=10\) until \(t=20\).}
    \label{fig:lorenz_streams}
\end{figure}

\subsubsection{Hidden coordinates}
We now consider the problem of identifying dynamics with a hidden variable that is never observed. In this case, the observations alternate between the first two coordinates on each time step, and the third is never observed. 
\cite{Ribera2022-nz} considered a very similar setup, introduced a variational annealing approach to identify dynamics in the presence of hidden variables, and also tested on variants of the Lorenz dynamics. \cite{hamzi-partial} approached the hidden variable problem using the lens of computational graph completion \cite{owhadi-CGC} and kernel learning, though only considered problems in discrete time without noise.
\begin{figure}
    \centering
    \includegraphics[width=1\linewidth]{  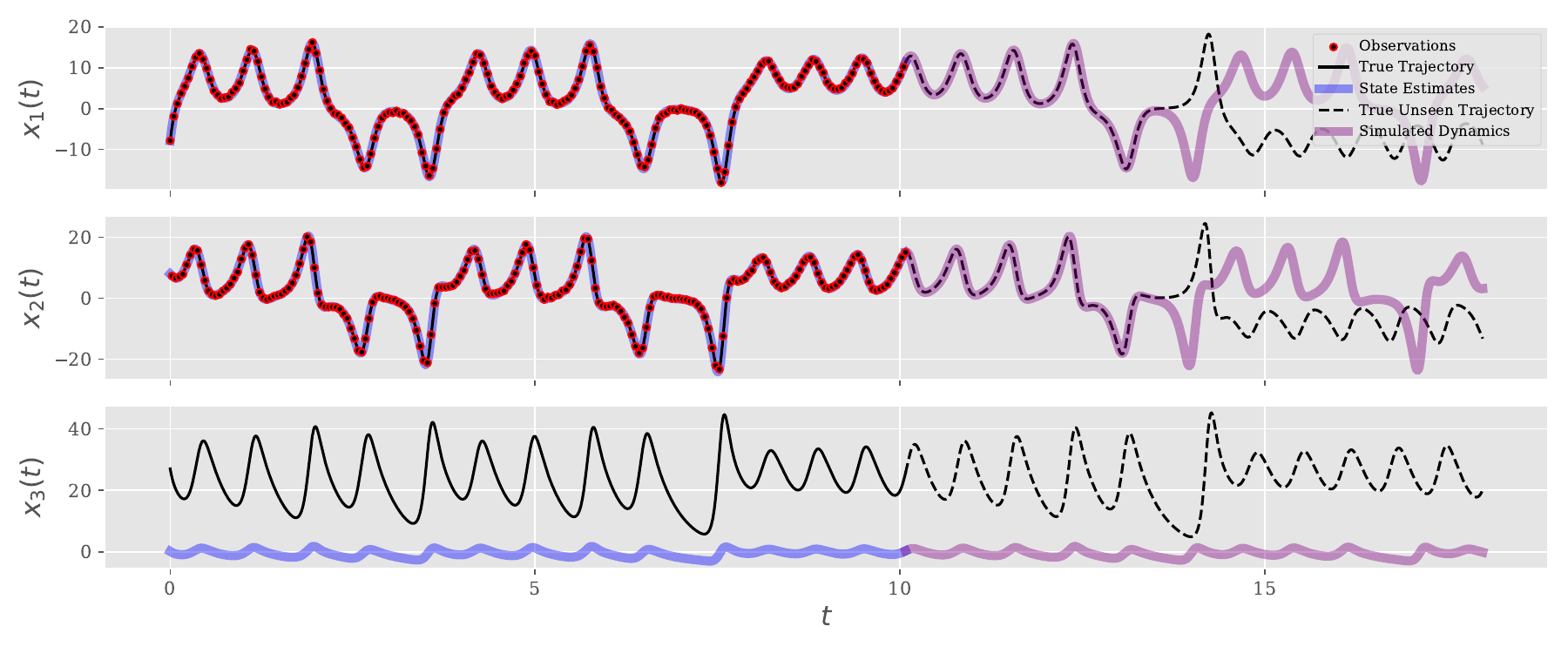}
    \caption{Results for Lorenz 63 system with an unobserved coordinate. Observations alternate between first two coordinates at sampling rate \(\Delta t = 0.025\) up to \(t = 10\), with no knowledge of the third, except that the system is known to be three dimensional. Our algorithm constructs a third, latent coordinate and corresponding dynamics which, when evolved together with the first two coordinates, reproduces the observed dynamics. Further, we can effectively extrapolate the time series forward by using our learned dynamics and imputing the value of our artificial coordinate as a new initial condition, generating good predictions for the first two coordinates. }
    \label{fig:lorenz_unobserved}
\end{figure}
In this case, we cannot expect to recover the values of the hidden third coordinate due to strong identifiability issues. For example, rescaling the third coordinate by an arbitrary scalar, and modifying the first two equations appropriately gives rise to an infinite family of equations of the same sparsity level which reproduce exactly the original dynamics, restricted to the first two coordinates. 

Nevertheless, we can learn an extra latent coordinate; when combined with the first two observed coordinates, these together satisfy a modified ODE that reproduces the dynamics on the observed coordinates. This is demonstrated by taking our state estimate at \(t =10\), the last observed data-point, and simulating the learned ODE forward in time to predict the evolution of the first two states. Doing so produces accurate trajectories for four additional cycles, including a lobe switch, before diverging from the true trajectory due to chaos. Note that modeling only the first two states with an ODE is insufficient--the flow would be self-intersecting and cannot satisfy any autonomous ODE without a third coordinate (no chaos in two dimensions). The results can be seen in \Cref{fig:lorenz_unobserved} and \cref{eqn:learned-lorenz-unobs}. Because the true hidden state is unidentifiable, we compute the state estimation error restricted to the first two coordinates and obtain \(\mse = 6.840\cdot 10^{-3}\). Similarly, since we cannot expect to recover the third coordinate, we omit the coefficient error. 

\begin{equation}
\begin{aligned}
\frac{d x_1}{d t} &= 0.722 \cdot 1 - 9.969\, x_1 + 9.936\, x_2 + 0.430\, x_3,\\
\frac{d x_2}{d t} &= -0.456 \cdot 1 - 0.453\, x_1 - 0.979\, x_2 - 0.316\, x_3 
                    + 0.408\, x_1 x_2 - 8.832\, x_1 x_3 + 0.161\, x_3^{2},\\
\frac{d x_3}{d t} &= -8.626 \cdot 1 + 0.057\, x_2 - 2.683\, x_3
                    + 0.133\, x_1 x_2 - 0.421\, x_1 x_3.
\end{aligned}
\label{eqn:learned-lorenz-unobs}
\end{equation}

Traditional approaches to handling partially observed systems and modeling time series often enrich the dynamics by inserting time delay coordinates \cite{martin2024TimeDelay,George2021Decomposing,brunton_chaos_2017,rand_detecting_1981,Bakarji2023-sq,Broomhead1989-bh,Lee2019-wf,Pan2020-nd}. However, that comes with its own challenges, such as selecting the number and spacing of delay coordinates, difficulties in continuous time, especially when observations are not evenly spaced, and introducing bias due to possibly noisy inputs. Our approach has its own drawback of requiring a separate state estimation problem to make predictions, and producing a somewhat uninterpretable third coordinate. We remark that it is possible to include time delay embeddings in our approach as well, essentially turning our collocation method for differential equations into a collocation approach for a \emph{delay} differential equations--we leave this to future work. 

\subsection{Higher order systems}
Our approach is also capable of learning higher order differential equations given only state measurements.  While devices exist that can directly measure first and second order derivatives (e.g. Doppler radar, gyroscopes, accelerometers), measurements of the system state are the most common, and second order laws are ubiquitous in physics. We consider the Van der Pol oscillator, a coupled version of the Duffing system, and the nonlinear pendulum under model misspecification.

\subsubsection{Van der Pol oscillator}

The Van der Pol oscillator is governed by the second-order differential equation
\begin{equation}
    \frac{d^2x}{dt^2} = \mu(1-x^2) \frac{dx}{dt} - x.
    \label{eqn:true-vdp}    
\end{equation}
We set \(\mu = 0.5\) and note that for \(\mu>0\), Van der Pol systems demonstrate an attractive limit cycle with fast and slow branches. We trained on equally sampled measurements with \(\Delta t=2\) for \(t \in [0,100]\) with additive Gaussian noise with variance $\sigma^2 = 0.1$. We took the feature library to contain cubic polynomials of $x$ and $\frac{d}{dt}x$. 
Our approach obtained the equation
\begin{equation}
        \frac{d^2 x}{dt^2} = -0.982 x+ 0.343 \frac{dx}{dt}                   -0.337 \frac{dx}{dt} x^2
        \label{eqn:vdp}
\end{equation}
with coefficient error \(\mae = 1.858 \cdot 10^{-1}\), and state estimation error \(\mse = 4.540\cdot 10^{-2}\). Our algorithm identifies nearly the correct equation with exactly the correct support.

\begin{figure}
    \centering
    \includegraphics[width=1\linewidth]{   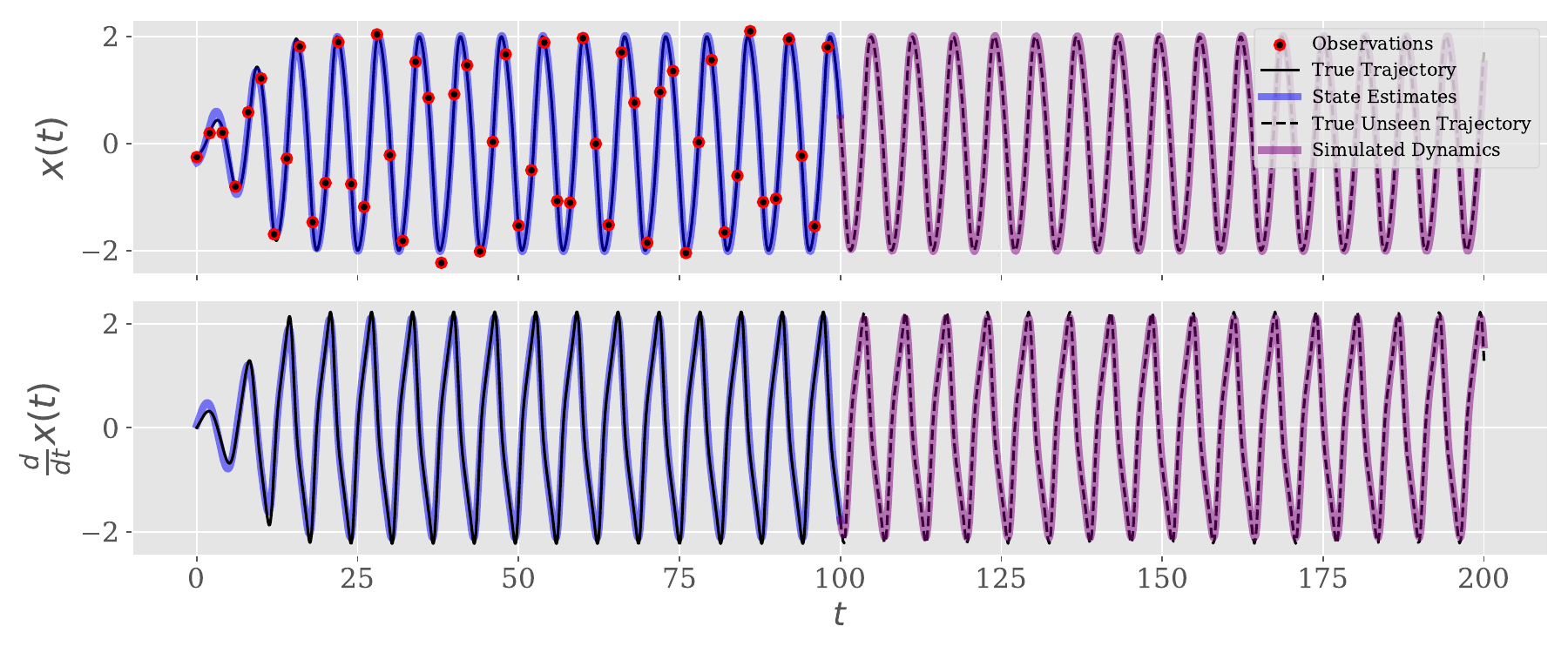}
    \caption{Results for the Van der Pol oscillator. Sampling is done only on system state, \(\vect{x}\), with sampling rate \(\Delta t = 2\) up to \(t=100\) with simulated dynamics from \(t=100\) to \(t=200\) using the learned equations from \cref{eqn:vdp}.}
    \label{fig:vdp}
\end{figure}

\begin{figure}
    \centering
    \includegraphics[width=1\linewidth]{   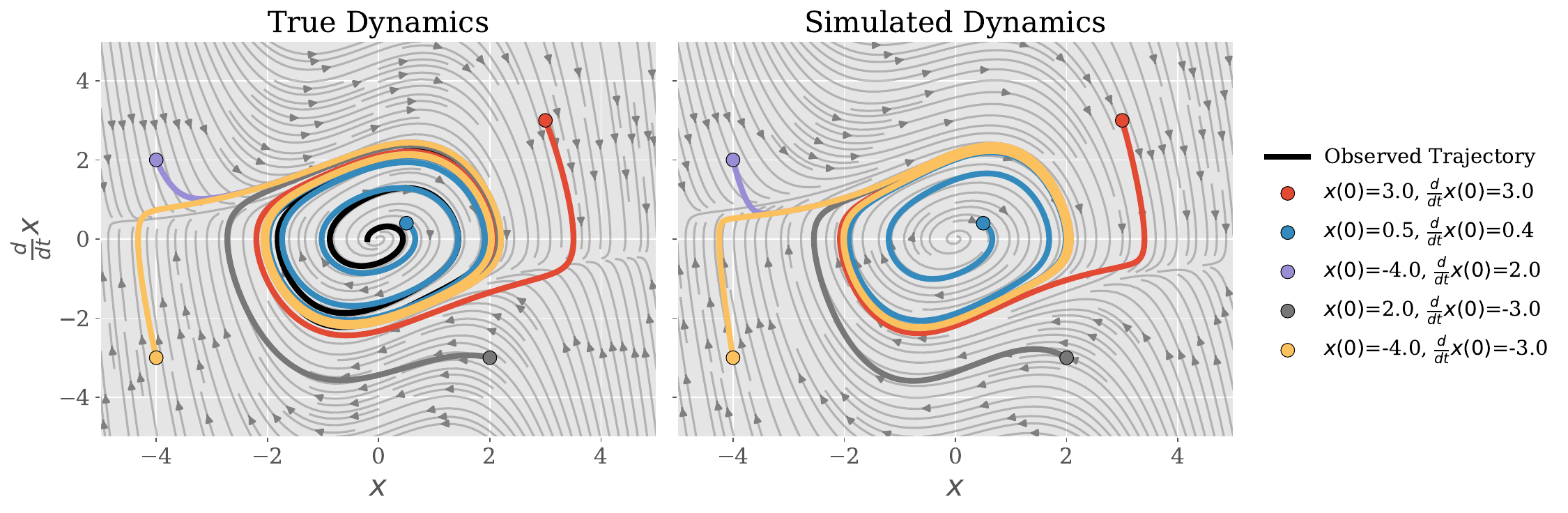}\\
    \caption{Left: True Van der Pol oscillator trajectories plotted from multiple initial conditions on phase portrait using \cref{eqn:true-vdp}. Right: Using learned model \cref{eqn:vdp} to simulate same trajectories provided same initial conditions.}  %The comparator models, on the other hand, lack this feature and instead allow some trajectories to diverge.}
    \label{fig:vdp-phase}
\end{figure}

\subsubsection{Coupled Duffing oscillators}
As an example of a larger second order ODE, we consider a ring of $d$ Duffing oscillators with nearest-neighbor coupling. For each dimension $i=1,\dots,d$, \(x_i(t)\) satisfies the differential equation
\begin{align}
    \frac{d^2}{dt^2} x_i(t) &= -\,\alpha_i\, x_i(t) - \beta_i\, x_i(t)^{3}
    + \mu\left(x_{i-1}(t)-2x_i(t)+x_{i+1}(t)\right),
\label{eqn:coupled-duffing}
\end{align}
with periodic indexing \(x_{d+1}:=x_1\). These dynamics are Hamiltonian. 
Set \(\vect{v} := \frac{d}{dt} \vect{x}\), and  
\begin{equation}\label{eqn:duffing-hamiltonian}
    \mathcal{H}(\vect{x},\vect{v})
    = \sum_{i=1}^d \frac{1}{2}\,v_i^2
    + \sum_{i=1}^d \left( \frac{\alpha_i}{2}\,x_i^2 + \frac{\beta_i}{4}\,x_i^4 \right)
    + \frac{\mu}{2}\sum_{i=1}^d \left(x_{i+1}-x_i\right)^2,
\end{equation}
the dynamics satisfy
\begin{equation}
    \frac{d}{dt}\vect{x} = \nabla_{\vect{v}} \mathcal{H}(\vect{x},\vect{v})
    ,\quad 
    \frac{d}{dt} \vect{v} = -\nabla_{\vect{x}} \mathcal{H}(\vect{x},\vect{v})
\end{equation}

We remark that the nonlinearity has a significant qualitative effect in the dynamics. If we consider a linear version where we drop the cubic restoring force, the dynamics would be governed by the matrix \(M:=\mathrm{Diag}((\alpha_i)_{i=1}^n) + \mu K\), where \(K\) is the graph Laplacian on a cycle. Taking \(\mu<0\) causes the coupling term to encourage exponential blow up, and with \(\mu\) negative enough relative to \(\alpha_i\), \(M\) will have both positive and negative eigenvalues. This implies unbounded level sets in the Hamiltonian, and divergent trajectories in the linear version.
However, the positive cubic restoring terms (\(\beta_i>0\)) are quartic in the Hamiltonian, which guarantees compact level sets, bounded trajectories \cite{teschl2012ordinary}.

We considered an example in $d=5$ dimensions with parameters \(\mu=-2,\; 
    \alpha_i \stackrel{\text{i.i.d.}}{\sim}\mathrm{Unif}[2,4],\; 
    \beta_i=2\). 
    % \label{eqn:duffing-params}
We initialized the system with all coordinates at the same position and velocity, \(x_i(0)=1\) and \(\frac{d}{dt} x_i(0)=0\) for \(i=1,\dots,d\).
We took equally spaced measurements of the position values \(x_i(t)\) with a sampling period \(\Delta t = 0.2\) between \(t = 0\) and \(t = 30\) as the input data, and include additive Gaussian noise of variance \(\sigma^2 = 0.04\) which corresponds to a signal to noise ratio of \(39.2\). 

Trajectory regularization, data weight, collocation penalties and points remain as in the previous example \(\lambda = 10^{-3}, \alpha = 1, \beta = 10^{5}, m=1000\). We took the sparsifier to be STLSQ with threshold 0.1, and ridge regularization parameter 0.01. The feature library includes monomials up to cubic order of \(x_i\) and \(\frac{d}{dt}x_i\), but excludes interaction terms,
\begin{equation}
    \dict= \left\{1,x_1,\frac{d}{dt}x_1,...,x_5,\frac{d}{dt}x_5,x_1^2,(\frac{d}{dt}x_1)^2,\dots,x_5^3,(\frac{d}{dt}x_5)^3\right\}.
\end{equation}
This library contains 31 functions (down from 286 if we were to include all 10-variable polynomials up to cubic order). Our algorithm correctly identifies the correct support in all five variables which amounts to finding the 20 active variables out of 155 with a maximum coefficient error of \(0.12\). 

\begin{figure}[ht]
  \centering
  \makebox[\textwidth][c]{%
    \includegraphics[width=1.22\textwidth]{   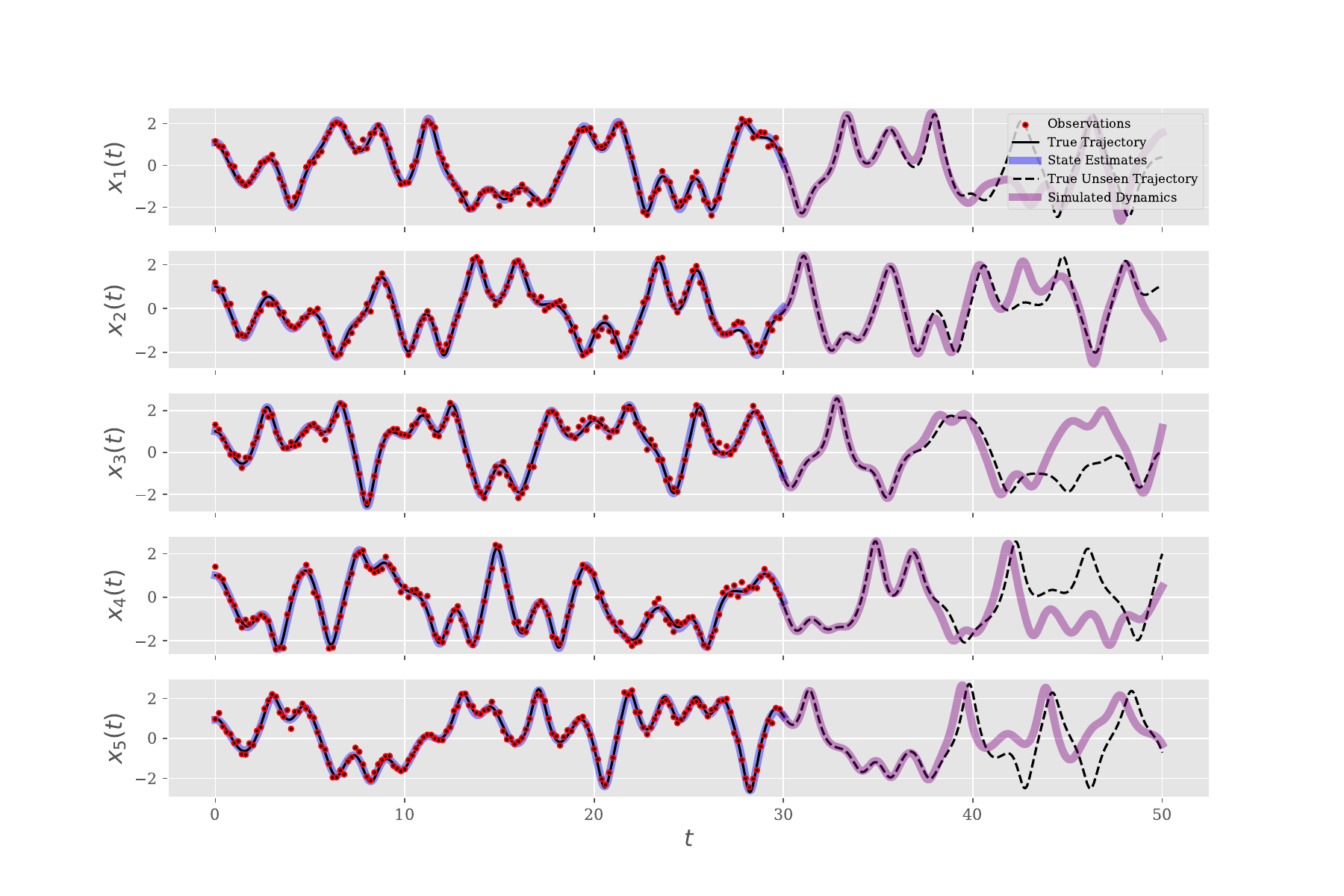}%
  }
  \caption{Position observations sampled at rate \(\Delta t=0.2\) up to \(t=30\) with added noise set to \(\sigma^2=0.04\) equating to a relative noise of \(2.55\%\). Learned and true dynamics from \cref{eqn:learned-cd} are simulated from \(t=30\) to \(t=50\).}
  \label{fig:cduff}
\end{figure}
We show the learned equations on the left, and the true equations on the right. 

\begin{equation}
\begin{aligned}
    \frac{d^2 x_1}{dt^2} &= 1.272\,x_1 - 2.074\,x_2 - 2.083\,x_5 - 1.890\,x_1^{3}, 
    & \frac{d^2 x_1}{dt^2} &= 1.392\,x_1 - 2\,x_2 - 2\,x_5 - 2\,x_1^{3} \\
    \frac{d^2 x_2}{dt^2} &= -1.951\,x_1 + 0.695\,x_2 - 2.010\,x_3 - 1.977\,x_2^{3}, 
    & \frac{d^2 x_2}{dt^2} &= -2\,x_1 + 0.711\,x_2 - 2\,x_3 - 2\,x_2^{3} \\
    \frac{d^2 x_3}{dt^2} &= -2.016\,x_2 + 1.862\,x_3 - 1.956\,x_4 - 2.045\,x_3^{3}, 
    & \frac{d^2 x_3}{dt^2} &= -2\,x_2 + 1.81\,x_3 - 2\,x_4 - 2\,x_3^{3} \\
    \frac{d^2 x_4}{dt^2} &= -2.023\,x_3 + 1.679\,x_4 - 2.072\,x_5 - 2.019\,x_4^{3}, 
    & \frac{d^2 x_4}{dt^2} &= -2\,x_3 + 1.521\,x_4 - 2\,x_5 - 2\,x_4^{3} \\
    \frac{d^2 x_5}{dt^2} &= -1.921\,x_1 - 1.961\,x_4 + 1.925\,x_5 - 2.054\,x_5^{3}, 
    & \frac{d^2 x_5}{dt^2} &= -2\,x_1 - 2\,x_4 + 1.988\,x_5 - 2\,x_5^{3}.
\end{aligned}
\label{eqn:learned-cd}
\end{equation}

We show estimates of the derivatives of one of the state variables and the error in the forward simulations of these learned dynamics starting from the last observation at \(t = 30\) in \Cref{fig:cduff} with \(\mse=3.554\cdot 10^{-2}\) and \(\mae = 3.627\cdot 10^{-2}\). Our model produces very accurate derivative estimates despite only having access to noisy observations of the original state variables. The regularized least squares collocation applied to a space of five times differentiable functions allows for derivative estimates up to fourth order to remain accurate, despite the dynamics being only second order. This level of performance is comparable to data simulation approaches where we know the model  \emph{beforehand} and apply high order nonlinear Kalman smoothing. The key point is that our approach simultaneously discovers the model and estimates the states.

\begin{figure}[ht]
  \centering
  \makebox[\textwidth][c]{%
    \includegraphics[width=1\textwidth]{ 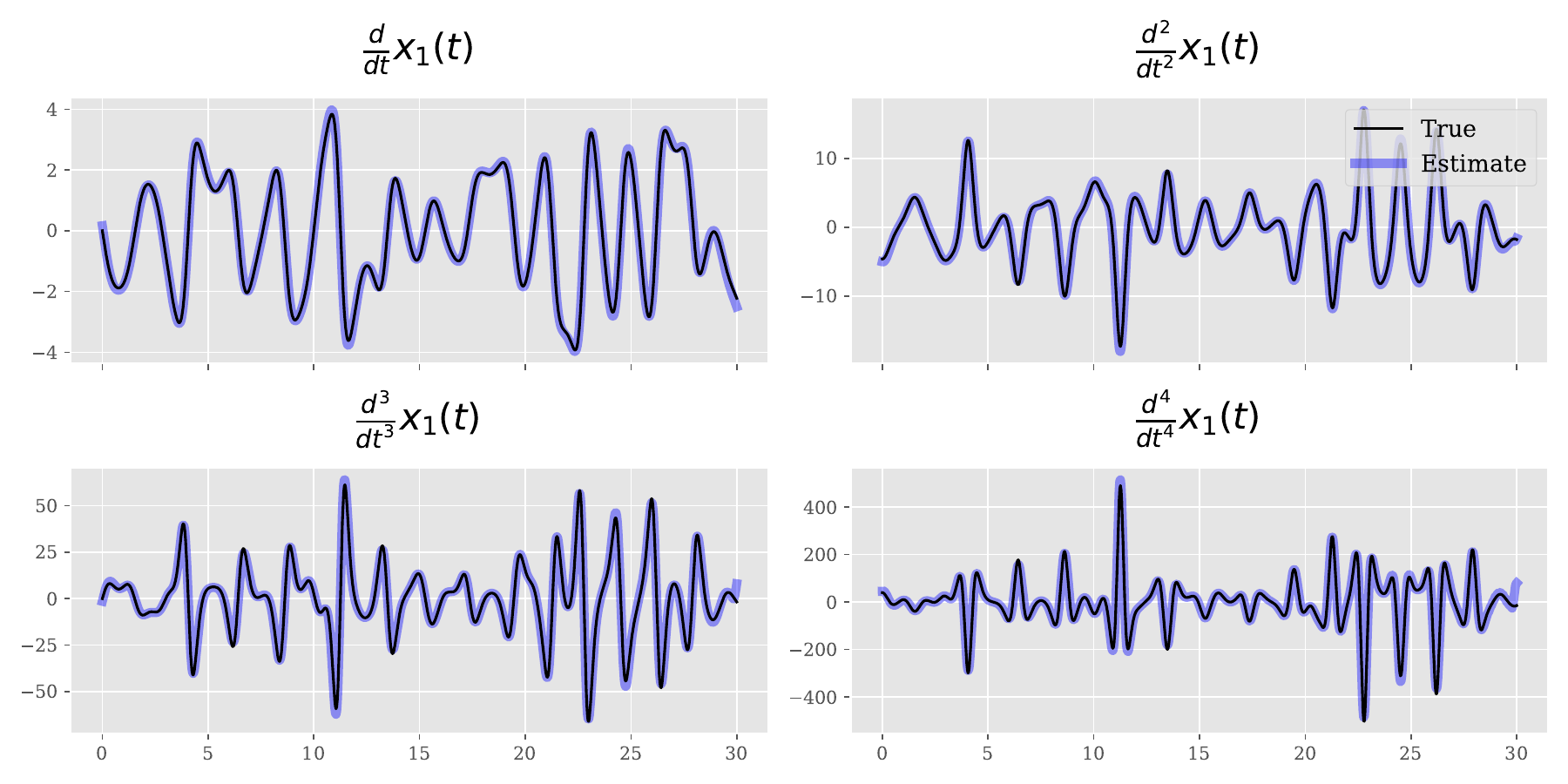}%
  }
  \caption{Results for estimating time derivatives of the coupled Duffing system. The relative mean square errors of the estimates of the first, second, third, and fourth derivatives are
  \(2.32\cdot 10^{-3}, 3.68\cdot 10^{-3}, 6.57\cdot 10^{-3}\) and \(1.1\cdot 10^{-2}\) respectively. In this case, our approach resolves the fine structure in estimating the state. 
  }
  % 0.0023194421487099655, 0.0036820986084760256,0.006565016586560378, and 0.011043487139330706.
  \label{fig:4deriv-cduff}
\end{figure}

\subsection{A misspecified model}\label{subsec:miss-spec}
We considered a case where the dynamics are misspecified by fitting polynomials to nonlinear pendulum dynamics. Our least squares relaxation of the dynamics allows some of the unmodeled dynamics to be accommodated as an error term, as in \cref{eq:control-version}. Robustness to misspecification is very useful; in many experimental or designed processes we may have little to no prior knowledge of the true functional form of the underlying dynamics. In these cases, we hope to maintain highly accurate state estimation (despite the misspecified dynamics), while producing a reasonable model for simulating the future.  

We consider the nonlinear pendulum,
\begin{equation}\label{eqn:nonlin-pend}
\frac{d^2x}{dt^2} = -\sin(x),
\end{equation}
and fit two models, one that is linear and another that is cubic. 
Simulated trajectories were created with initial displacement of $x(0) = 3$ and no initial velocity $\frac{d}{dt}x(0) = 0$.  Observations were made with a sampling period of \(\Delta t = 1\) over the interval \(t\in [0,100]\), with Gaussian noise ($\sigma^2 = 0.01$). 

The relatively large initial condition \(x(0) = 3\) means that the nonlinearity has a significant effect; replacing \(\sin(x)\) with a cubic Taylor series centered around zero results in finite time blowup with these initial conditions. 
We obtained the following models:
\begin{align}
    \text{Linear:}\quad 
    \frac{d^2x}{dt^2} &= -0.151\,x
    \label{eq:linear}\\[4pt]
    \text{Cubic:}\quad 
    \frac{d^2x}{dt^2} &= -0.796\,x + 0.086\,x^3.
    \label{eq:cubic}
\end{align}
Both of these perform very well in state estimation, despite only forming a somewhat crude approximation to the dynamics as evaluated through forward simulation (see \Cref{fig:nonlin-pend}). The filtering error of the linear model was \(\mse = 1.250\cdot 10^{-1}\)
and for the cubic model, was \(\mse = 2.221\cdot 10^{-2}\). 
The linear model matched the period of the true nonlinear dynamics \cref{eqn:nonlin-pend}, but did not have enough degrees of freedom to match the amplitude and misrepresented the shapes of the peaks. The cubic model performed quite well, especially in state estimation and matching the shapes of the oscillations, but has a slightly incorrect period. 
\begin{figure}
    \centering
    \includegraphics[width=
    1\linewidth]{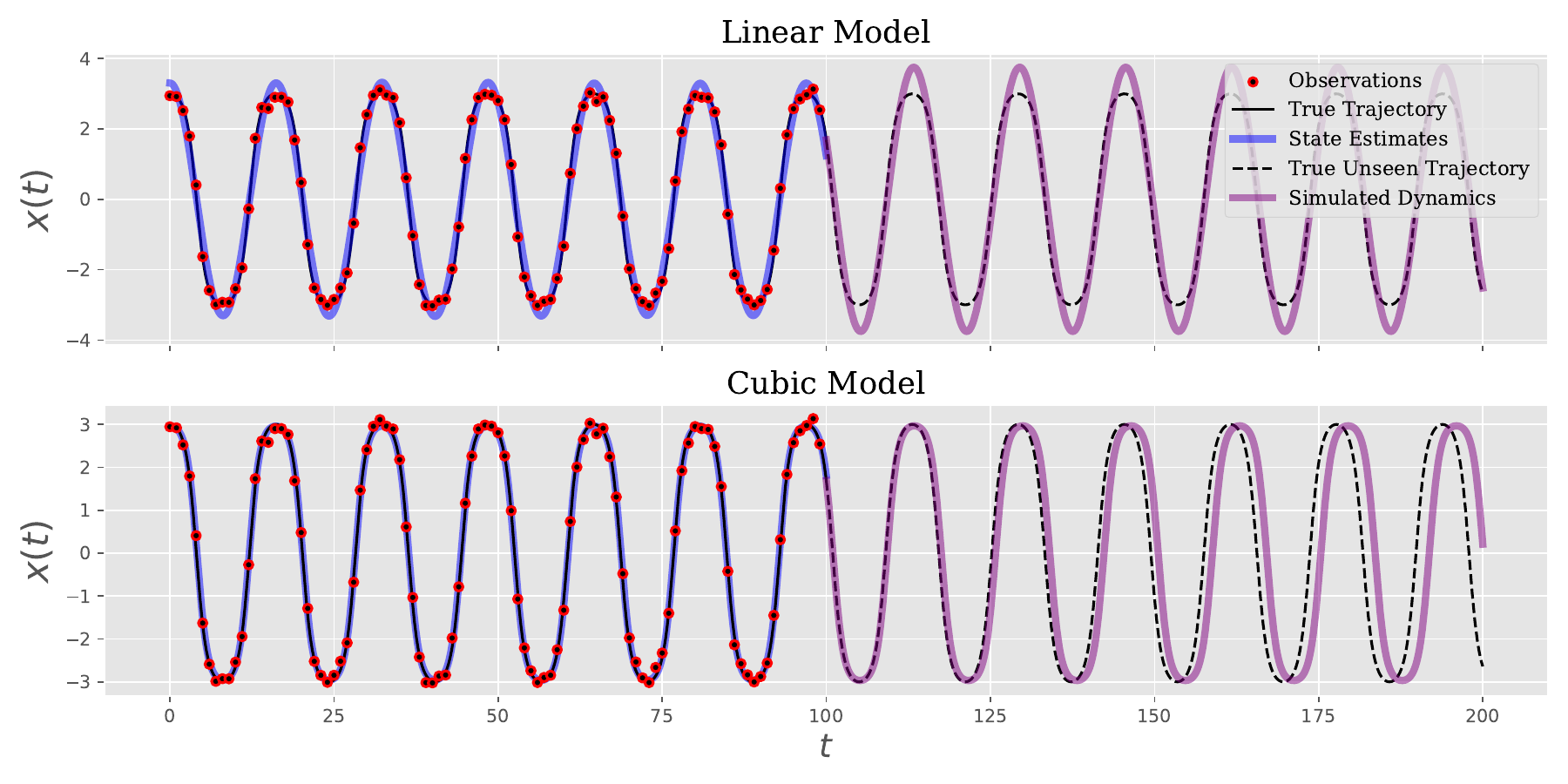}
    \caption{ Results for misspecified linear and cubic models for nonlinear pendulum. Observations were sampled at a rate \(\Delta t= 1\) up to \(t=100\) with added noise with variance \(\sigma^2=0.01\). Learned dynamics for both linear and cubic model, \cref{eq:linear} and \cref{eq:cubic}, are simulated from \(t=100\) to \(t=200\) where the linear model successfully captures periodicity of the true trajectory, but fails to capture the correct amplitude, whereas this is flipped for the cubic model.
}
    \label{fig:nonlin-pend}
\end{figure}
\subsection{Benchmarking}\label{subsec:comparisons}
We now consider two systematic experiments testing the robustness of our method to noise, and comparing primarily to ODR-BINDy \cite{Fung2025-odr}. 

\subsubsection{Damped nonlinear oscillator}\label{subsubsec:nonlin-osc}
We considered increasing noise levels with fixed trajectory length and observation frequency for a nonlinear damped oscillator that was studied in \cite{Fung2025-odr}. The dynamics are given by
\begin{equation}
\begin{aligned}
    \frac{d}{dt}x_1 &= -0.1 x_1^3 + 2x_2^3 \\
    \frac{d}{dt}x_2 &= 2x_1^3 -0.1 x_2^3. 
\end{aligned}
\end{equation}
We took the initial conditions \(\vect{x}(0) = (2,0)\), and took samples of both variables with sampling period \(\Delta t = 0.05\) for \(t \in [0,10]\). We first considered an example with noise level \(\sigma = 0.32\). We set trajectory regularization \(\lambda = 10^{-5}\), data weight \(\alpha = 10\), collocation penalty \(\beta = 10^5\), and used STLSQ with threshold \(0.08\) and ridge parameter \(100\). This is a surprisingly challenging example, especially with this level of noise--our model does not identify the correct support, but still performs quite well in estimating the state. We obtained \(\mse = 6.60\times 10^{-2}\), \(\mae = 2.47\times 10^{-1}\), and the equations
\begin{equation}
    \begin{split}
        \frac{d}{dt}x_1 &= 0.078\, x_2 - 0.560\, x_1 x_2^{2} + 1.859\, x_2^{3} \\[6pt]
        \frac{d}{dt}x_2 &= - 0.008\, x_1 - 0.212\, x_1^{2} - 1.950\, x_1^{3} - 0.284\, x_1 x_2^{2}.
    \end{split}
    \label{eqn:learned-osc}
\end{equation}
The oscillatory terms were recovered with reasonable accuracy, but the small damping was hard to capture correctly. From \Cref{fig:extra-osc}, we can see that the state estimates and reproduced dynamics are still close, which suggests some ill-posedness in the sparse recovery problem.

Next, we tested on a sequence of noise levels \(\sigma = 0.02,0.04,0.08,0.16,0.32,0.64\) and averaged the values of \(\mse\) and \(\mae\) over \(32\) random seeds. In this case, we found that including ensembling \cite{fasel_ensemble-sindy_2022} in the inner sparsifier greatly improved the consistency of the results for higher noise levels. We compared five models, weak SINDy\cite{Messenger2021a}, weak SINDy with ensembling, JSINDy, JSINDy with ensembling, and ODR-BINDy \cite{fasel_ensemble-sindy_2022,Fung2025-odr}. We used PySINDy for ensembling and weak SINDy \cite{Kaptanoglu2022}, and the GitHub repository\footnote{\url{https://github.com/llfung/ODR-BINDy}} associated to \cite{Fung2025-odr} for ODR-BINDy.  We also tested with other two-step variants of SINDy, but they all performed worse than weak SINDy, so we omit them. As ODR-BINDy also produces state estimates, we compare the state estimation error between JSINDy, JSINDy with ensembling and ODR-BINDy. 

We note that our ensembling only occurs inside of the sparsifier \(\mathcal{S}\), and therefore does not incur much additional computational cost, which is dominated by the nonlinear least squares problems rather than the sparsifying iterations. In cases where the model is already successful, ensembling can introduce additional extraneous features, but for challenging problems, it seems to improve the robustness. Perhaps surprisingly, the error of weak ensemble SINDy was quite high for low levels of noise, but did not degrade as noise was increased for the scale that we considered.

In terms of the coefficient error, and without ensembling the sparsifier, our model performs noticeably worse than ODR-BINDy; maximizing the Bayesian evidence seems to be a good model selection criterion. Upon introducing ensembling, our model selection stabilizes and the coefficient errors become comparable. In terms of state estimation, JSINDy both with and without ensembling performs well across varying levels of noise, regardless of the coefficient accuracy, owing to the stabilizing effect of the RKHS regularization and least squares relaxation.

\begin{figure}
  \centering
  \makebox[\textwidth][c]{
    \includegraphics[width=1.1\textwidth]{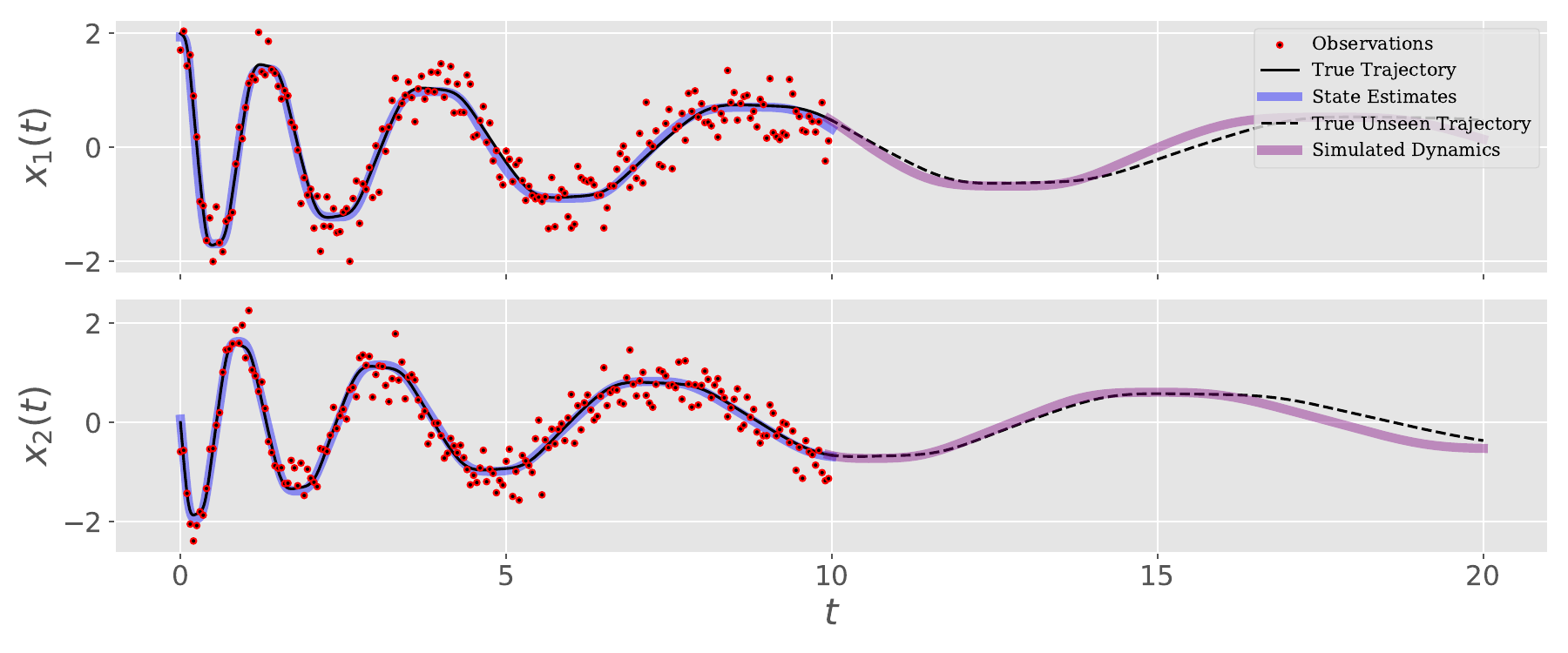}
  } 
  \caption{Results for the nonlinear damped oscillator with full state observations, \(\Delta t = 0.05\) in the interval \([0,10]\), and \(\sigma = 0.32\) additive noise. While the learned equations in \cref{eqn:learned-osc} do not match the true dynamics in terms of the coefficient values, JSINDy effectively infers the states, and produces a reasonable extrapolation in forward simulation.}
  \label{fig:extra-osc}
\end{figure}

\begin{figure}
  \centering
  \makebox[\textwidth][c]{%
    \includegraphics[width=1.1\textwidth]{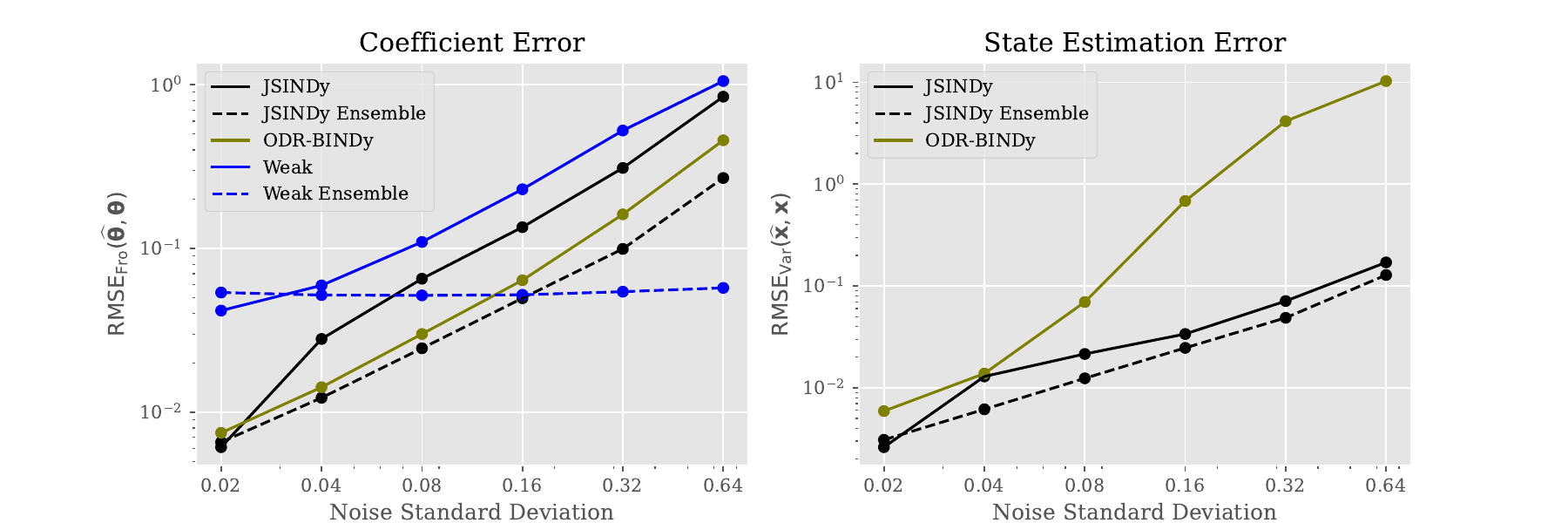}%
  }
  \caption{Results for the damped nonlinear oscillator across varying levels of noise. Left: coefficient error \(\mae\), Right: state estimation error \(\mse\).  In this case, ensembling significantly improves the coefficient error for both weak SINDy and JSINDy. For JSINDy, including ensembling inside of the sparsifier brings the coefficient error from noticeably worse than ODR-BINDy to slightly better. Overall, JSINDy performs very well in state estimation across varying levels of noise.}
  \label{fig:nonlin-osc}
\end{figure}

\subsubsection{Lorenz 63 example}
\label{subsubsec:lorenz-bench}
We also tested the performance of our method on varying levels of noise and durations of observations with fixed observation frequency on the Lorenz 63 dynamics \eqref{eqn:true_lorenz} with results in \Cref{fig:lorenz_benchmark}. We applied JSINDy with data weight \(\alpha = 1\), collocation weight \(\beta = 10^5\), regularization \(\lambda = 10^{-5}\), and used SR3 as a sparsifier with an L1 penalty with weight 5, and relaxation coefficient \(\nu = 0.05\) \cite{SR3}. We compare to ensemble weak SINDy and ODR-BINDy. The setup is identical to that of \cite[Sec. 3.2]{Fung2025-odr}, though we did not exclude the constant function from the library for weak SINDy. ODR-BINDy performed best in variable selection, likely owing to their approach of Bayesian evidence maximization. While also better than JSINDy in coefficient error, the difference is somewhat less pronounced. A much larger difference can be seen between JSINDy and ensemble weak SINDy, as well as the other models shown in \cite[Sec. 3.2]{Fung2025-odr}.  
\begin{figure}
    \centering
    \includegraphics[width=1\linewidth]{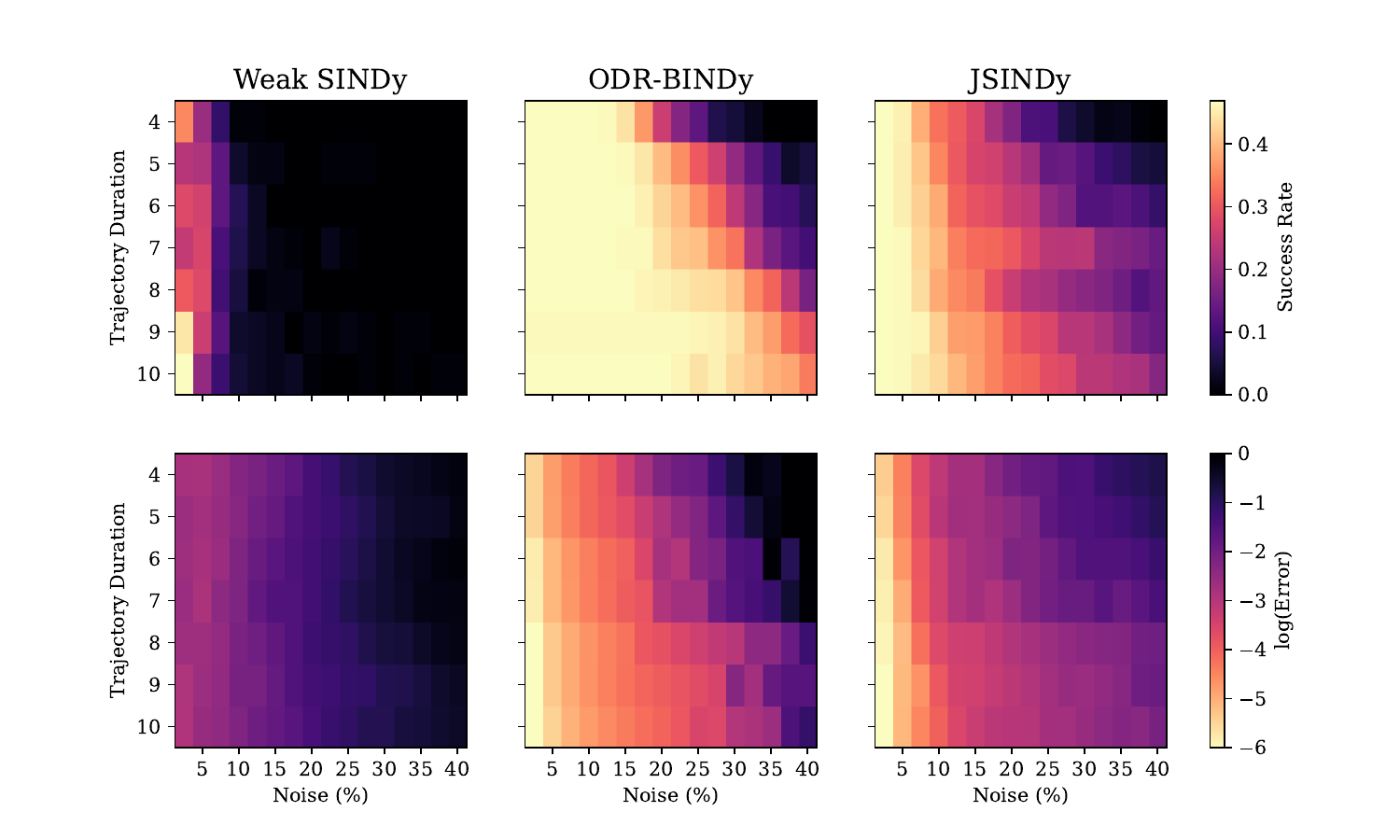}
    \caption{Experiment across different trajectory lengths and noise levels, averaged over 128 independent trials and \(\Delta t = 0.01\). \textbf{Top:} Proportion of trials where the correct support was exactly recovered \textbf{Bottom:} Relative coefficient error \(\mae\) of identified system. }
    \label{fig:lorenz_benchmark}
\end{figure}

\section{Discussion and conclusion}\label{sec:discussion}

A key technique used in this work is the least squares relaxation of the ODE constraint:
\begin{equation}
    \dtp \vect{x} = f(\xdiffs(t);\vect{\theta})
\end{equation}
into the penalty term 
\begin{equation}
    \LL_C(\vect{x},\vect{\theta}) = \frac{1}{2}\int_0^T \left|\dtp \vect{x} - f(\xdiffs(t);\vect{\theta})\right|^2 dt. 
\end{equation}

In this section, we provide some insight into the mathematical underpinnings of this approach, discuss the theoretical advantages, and highlight the role of regularization, while providing additional context and reference to trace many of these ideas in the literature.

The least squares relaxation has been applied in the PDE inverse problem setting \cite{tristan_pde}, where it was found to mitigate the effects of local minima \cite{vanLeeuwen2013Mitigating}. It was further analyzed in \cite{tristan_analysis}, which presented a representer theorem for relaxations of parametric linear PDEs and cast the inversion problem as minimizing a special RKHS norm. The least squares approach is also common in the ML literature, most notably through physics-informed neural networks \cite{Raissi2019-so}, as well as approaches to learning differential equations and parameter estimation \cite{sun2022bayesiansplinelearningequation,Chen2021,ijcai2021p283,Fung2025-odr,jalalian2025dataefficientkernelmethodslearning,rudy_soft}. 
In our case, the additional error associated to the ODE from applying a least-squares approach can be beneficial to the overall learning process, rather than a limitation. 
We prefer the relaxation approach to alternatives, such as reduced or shooting methods that view the ODE as a constraint and enforce feasibility from the start, which can be seen in neural ODEs, inverse problems, and optimal control \cite{kidger2021on,Chen2018,Givoli2021-bc,Lions1971}, or formulations that enforce the hard constraint within an all-at-once formulation such as \cite{Hokanson2023-sr}, which optimizes by solving the discretized KKT conditions.

The RKHS regularized least squares approach that we take has two major advantages over reduced approaches. First, state estimates for any length of time are well defined for a much broader class of ODEs, which can be seen from the existence of bounded solutions in \Cref{thm:cont-existence}. Further, the state estimates are allowed to be self-correcting in that the additional data misfit terms may inform the state estimates. In contrast, in reduced approaches, certain parameter values may lead to finite time divergence of the ODE. 
Second, this approach improves the regularity of the objective with respect to the system parameters \(\vect{\theta}\). For chaotic dynamics, we can generically expect that small perturbations to parameters of the system will lead to exponentially large changes over long periods of time relative to the Lyapunov exponent. 
Several works have observed that this leads to extremely ill-conditioned and non-convex objectives in reduced formulations, especially in chaotic regimes where the mismatch between simulated and observed trajectories can grow exponentially with time \cite{ribeiro2020smoothness,Chakraborty2024multistep,Carbonell2015multishooting}. 
As a result, gradient-based optimization often stalls, converges to spurious local minima, or fails entirely. These issues can be mitigated significantly through multiple-shooting and multipoint-penalty methods \cite{Chakraborty2024multistep,Carbonell2015multishooting,Chung2022-hz}, which break a long trajectory into shorter segments, while simulating within each segment. The full space approaches considered in this work can, in a sense, be seen as the limit of many segments in a multiple-shooting method, but we leave a more systematic comparison to future work.

The least squares relaxation avoids simulating forward trajectories altogether and instead distributes the residual over time, substantially smoothing the optimization landscape and improving robustness. Further, it makes the objective a \emph{linear} least squares regression problem in the coefficients \(\vect{\theta}\) for any fixed state estimate \(\vect{x}\). This natural perspective also appears in \cite{tristan_analysis} in the context of recovering a coefficient field in a PDE-constrained inverse problem. \cite{Fung2025-odr} points out the connection to orthogonal distance regression in all-at-once approaches, casting  optimization with respect to the state variables as adjusting the inputs to correct an error-in-variables bias. Indeed, in the linear case, such an approach appears as early as 1950 in Householder's method for exponential fitting \cite{householder1950prony}, where the error-in-variables issue in Prony's method is corrected by jointly fitting state estimates and discrete time dynamics. 

Least squares relaxation of the ODE constraint can also be seen as minimizing a backwards error of the ODE. Introducing an auxiliary forcing function \(\vect{r}(t)\), we may equivalently write the least squares penalty as
\begin{equation}\label{eq:control-version}
\LL_C(\vect{x},\vect{\theta}) =
\left\{
\begin{alignedat}{2}
& \min_{\vect{r}\in L^2([0,T],\R^d)} && \int_0^T |\vect{r}(t)|^2 \, dt \\
& \text{subject to }  \dtp  &&\vect{x}(t) = f(\xdiffs(t);\vect{\theta}) + \vect{r}(t)
\end{alignedat}
\right.,
\end{equation}
where the ODE is understood in the weak sense. 

When this formulation of \(\mathcal L_C\) is embedded into the optimization problem \cref{eqn:jsindy-opt-sparse} with an appropriate lifting to optimize with respect to \(\vect{r}\), this can be interpreted as finding a control or forcing term $\vect{r}(t)$ with small \(L^2\) norm which is optimized to steer \(\vect{x}\) to match the data and have small RKHS norm. This view is also well motivated when the model is misspecified, where \(\vect{r}\) can act as an error term in the dynamics. 

While one can control the error in terms of the size of the residual \(\vect{r}\) and a Lipschitz assumption on \(f\) via classical error analysis for ODEs using Gr\"{o}nwall's inequality, this will generally result in overly pessimistic error estimates which grow exponentially in time as it assumes worst case behavior of \(\vect{r}\). Rather, \(\vect{r}\) largely pushes the trajectory to match the data, so its effects should be small when the model is well specified, and large in misspecified cases, where it is necessary. By assimilating the observation data over a long trajectory, the error can be compared to that of Kalman filtering/smoothing (see e.g.~\cite{aravkin2017generalized}) under very low system noise and a reasonably accurate system model. 

In our case, the regularity of \(\vect{r}\) follows from the smoothness of \(\vect{x}\) lent by the RKHS regularization and continuity of \(f\), noting that \(\vect{r}(t)  = \dtp\vect{x}(t)-f(\xdiffs(t);\vect{\theta})\), the difference of two continuous functions. However, without regularization of \(\vect{x}\), there is no reason to believe that \(\vect{r}\) is continuous, so \(\vect{x}\) will only be piecewise \(p\)-times differentiable; for first order ODEs, one should expect kinks in the state estimates without higher order regularization. 
The closely related approach of ODR-BINDy \cite{Fung2025-odr} interprets the least squares approach as a stochastic relaxation, modeling error in the ODE as independent white noise. However, it can be difficult to rigorously interpret the noise stochastically taking the function space limit of inferring \(\vect{x}(\cdot)\); in the random setting and without regularization, the covariance of the residuals becomes increasingly important as the mesh is refined and \(\vect{r}(t)\) is considered as a random process. In the case that the residuals are uncorrelated (e.g.  \(\vect{r}(t),\vect{r}(t')\) are independent for \(t\neq t'\) ), then either the pointwise variances have to blow up, or the noise completely washes out under the integral. The probabilistic model can be seen as applying maximum a posteriori estimation to the stochastic differential equation 
\begin{equation}
    d\vect{x} = f(\vect{x})dt + \epsilon d\vect{B}_t.
\end{equation}
In this model, \(\frac{d}{dt} \vect{x}\) will not be defined pointwise, possessing only the regularity of white noise. Of course, these considerations do not manifest computationally; the error terms end up correlated, there is no ``real" Brownian motion happening, and most reasonable discretizations will implicitly regularize the trajectory estimate. However, this does raise the question of determining a computationally effective and complete stochastic model in the vein of probabilistic numerics \cite{probnum}, as well as the idea of penalizing the residual in a stronger norm than \(L^2\). 
\subsection{Conclusion}
We proposed and demonstrated the effectiveness of JSINDy in both system identification and filtering by performing both jointly. We addressed a number of challenging problems in continuous time system identification within a single framework based on collocation in RKHSs, including scarce sampling, far below what has been considered in previous works, partially observed states, and higher order dynamics. 
Building on other recent advancements in all-at-once system identification methods \cite{Fung2025-odr,Chen2021,ijcai2021p283,Ribera2022-nz,Hokanson2023-sr}, this work provides more strong evidence for the superiority of all-at-once methods for system identification, and demonstrates significant flexibility. In short, the assumption that the data satisfy some dynamics is a very good prior to use when it is true. 
Overall, the RKHS formulation provides a flexible and effective tool for modeling while providing a principled approach to regularization, resulting in stable optimization, and allowing for seamlessly incorporating prior assumptions, data, and dynamics.

Currently, the primary drawback of JSINDy is that dense kernel matrices arise from the global approximation, leading to cubic complexity in the trajectory length. This can, however, be mitigated through a state-space representation of Mat\'ern Gaussian processes as in \cite{Saa12}, which essentially amounts to a finite difference approximation to the associated Sobolev norm, and would result in banded Jacobian matrices. 

Further improvements are available on the side of optimization. The objective function in \cref{eqn:jsindy-opt} is amenable to partial minimization \cite{Aravkin2024-bz,O-Leary2013-mg,Golub1973-el} with respect to the coefficients \(\vect{\theta}\) in solving the fixed-support nonlinear least squares problems. Our approach of alternating with an algorithmic sparsifier is another choice which could use more investigation. In particular, incorporating Bayesian evidence as a model selection criteria could be highly effective, while raising interesting questions for efficiently solving the discrete optimization problems that arise \cite{Fung2025-odr,Fung2025-aa,Klishin2025-pl}. 

\section*{Acknowledgments}
AH and BH acknowledge support from the National Science Foundation under awards 
2208535 (Machine Learning for Bayesian Inverse Problems) and 
2337678 (CAREER: Gaussian Processes for Scientific Machine
Learning: Theoretical Analysis and Computational Algorithms). AH acknowledges support from a Carl E. Pearson Fellowship. JMSH was partially supported by the Veterans Administration via the Post-9/11 GI Bill.
JNK was
supported in part by the NSF AI Institute for Dynamical Systems (dynamicsai.org), grant 2112085. The numerical experiments involving ODR-BINDy were completed on Hyak, UW’s high performance computing cluster, which is funded by the UW student technology fee. 

% \uspunctuation
\printbibliography

@article{martin2024TimeDelay,
    author = {Martin, R. S. and Greve, C. M. and Huerta, C. E. and Wong, A. S. and Koo, J. W. and Eckhardt, D. Q.},
    title = {A robust time-delay selection criterion applied to convergent cross mapping},
    journal = {Chaos: An Interdisciplinary Journal of Nonlinear Science},
    volume = {34},
    number = {9},
    pages = {093110},
    year = {2024},
    month = {09},
    issn = {1054-1500},
    doi = {10.1063/5.0209028},
    url = {https://doi.org/10.1063/5.0209028},
    eprint = {https://pubs.aip.org/aip/cha/article-pdf/doi/10.1063/5.0209028/20143851/093110_1_5.0209028.pdf},
}

@book{teschl2012ordinary,
  title={Ordinary differential equations and dynamical systems},
  author={Teschl, Gerald},
  volume={140},
  year={2012},
  publisher={American Mathematical Soc.}
}

@article{wright1999numerical,
  title={Numerical optimization},
  author={Wright, Stephen and Nocedal, Jorge and others},
  journal={Springer Science},
  volume={35},
  number={67-68},
  pages={7},
  year={1999}
}

@article{levenberg,
  title={A method for the solution of certain non-linear problems in least squares},
  author={Levenberg, Kenneth},
  journal={Quarterly of applied mathematics},
  volume={2},
  number={2},
  pages={164--168},
  year={1944}
}

@article{marquardt,
author = {Marquardt, Donald W.},
title = {An Algorithm for Least-Squares Estimation of Nonlinear Parameters},
journal = {Journal of the Society for Industrial and Applied Mathematics},
volume = {11},
number = {2},
pages = {431-441},
year = {1963},
doi = {10.1137/0111030},
URL = {
        https://doi.org/10.1137/0111030
},
eprint = { 
    
        https://doi.org/10.1137/0111030
}

}

@article{Klishin2025-pl,
  title = {Statistical mechanics of dynamical system identification},
  author = {Klishin, Andrei A. and Bakarji, Joseph and Kutz, J. Nathan and Manohar, Krithika},
  journal = {Phys. Rev. Res.},
  volume = {7},
  issue = {3},
  pages = {033181},
  numpages = {16},
  year = {2025},
  month = {Aug},
  publisher = {American Physical Society},
  doi = {10.1103/4d98-tdlp},
  url = {https://link.aps.org/doi/10.1103/4d98-tdlp},
  number = {3},
}

@article{aravkin2017generalized,
  title={Generalized Kalman smoothing: Modeling and algorithms},
  author={Aravkin, Aleksandr and Burke, James V and Ljung, Lennart and Lozano, Aurelie and Pillonetto, Gianluigi},
  journal={Automatica},
  volume={86},
  pages={63--86},
  year={2017},
  publisher={Elsevier}
}

@ARTICLE{Fung2025-aa,
  title        = {{Rapid Bayesian identification of sparse nonlinear dynamics
                  from scarce and noisy data}},
  author       = {Fung, Lloyd and Fasel, Urban and Juniper, Matthew},
  journaltitle = {Proc. Math. Phys. Eng. Sci.},
  publisher    = {The Royal Society},
  volume       = {481},
  issue        = {2307},
  date         = {2025-02},
  urldate      = {2025-11-17},
  language     = {en},
  journal = {Proc. Math. Phys. Eng. Sci.},
  year = {2025},
  number = {2307},
}

@ARTICLE{O-Leary2013-mg,
  title        = {{Variable projection for nonlinear least squares problems}},
  author       = {O'Leary, Dianne P and Rust, Bert W},
  journaltitle = {Comput. Optim. Appl.},
  publisher    = {Springer Science and Business Media LLC},
  volume       = {54},
  issue        = {3},
  pages        = {579--593},
  date         = {2013-04},
  urldate      = {2025-11-17},
  language     = {en},
  journal = {Comput. Optim. Appl.},
  year = {2013},
  number = {3},
}

@ARTICLE{Golub1973-el,
  title        = {{The differentiation of pseudo-inverses and nonlinear least
                  squares problems whose variables separate}},
  author       = {Golub, G H and Pereyra, V},
  journaltitle = {SIAM J. Numer. Anal.},
  publisher    = {Society for Industrial \& Applied Mathematics (SIAM)},
  volume       = {10},
  issue        = {2},
  pages        = {413--432},
  date         = {1973-04},
  urldate      = {2025-11-17},
  language     = {en},
  journal = {SIAM J. Numer. Anal.},
  year = {1973},
  number = {2},
}

@phdthesis{Saa12,
	author = {Saat{\c{c}}i, Yunus},
	school = {University of Cambridge},
	title = {Scalable inference for structured Gaussian process models},
	year = {2012}}

@BOOK{probnum,
  title     = {{Probabilistic numerics: Computation as machine learning}},
  author    = {Hennig, Philipp and Osborne, Michael A and Kersting, Hans P},
  publisher = {Cambridge University Press},
  location  = {Cambridge, England},
  date      = {2022-06-30},
  pagetotal = {410},
  urldate   = {2025-11-17},
  language  = {en},
  year = {2022},
}

@ARTICLE{owhadi-CGC,
  title        = {{Computational graph completion}},
  author       = {Owhadi, Houman},
  journaltitle = {Res. Math. Sci.},
  publisher    = {Springer Science and Business Media LLC},
  volume       = {9},
  issue        = {2},
  pages        = {1--33},
  date         = {2022-06-18},
  urldate      = {2025-10-07},
  language     = {en},
  journal = {Res. Math. Sci.},
  year = {2022},
  number = {2},
}

@article{hamzi-partial,
title = {Learning dynamical systems from data: A simple cross-validation perspective, part IV: Case with partial observations},
journal = {Physica D: Nonlinear Phenomena},
volume = {454},
pages = {133853},
year = {2023},
issn = {0167-2789},
doi = {https://doi.org/10.1016/j.physd.2023.133853},
url = {https://www.sciencedirect.com/science/article/pii/S0167278923002075},
author = {Boumediene Hamzi and Houman Owhadi and Yannis Kevrekidis},
}

@ARTICLE{SR3,
  author={Zheng, Peng and Askham, Travis and Brunton, Steven L. and Kutz, J. Nathan and Aravkin, Aleksandr Y.},
  journal={IEEE Access}, 
  title={A Unified Framework for Sparse Relaxed Regularized Regression: SR3}, 
  year={2019},
  volume={7},
  number={},
  pages={1404-1423},
  doi={10.1109/ACCESS.2018.2886528}}

@misc{jsindyRepo,
  title        = {{JSINDy}},
  year         = {2025},
  howpublished = {\url{https://github.com/AHsu98/jsindy}},
}

@ARTICLE{Varah1982-ay,
  title        = {{A spline least squares method for numerical parameter
                  estimation in differential equations}},
  author       = {Varah, J M},
  journaltitle = {SIAM J. Sci. Stat. Comput.},
  publisher    = {Society for Industrial \& Applied Mathematics (SIAM)},
  volume       = {3},
  issue        = {1},
  pages        = {28--46},
  date         = {1982-03},
  journal = {SIAM J. Sci. Stat. Comput.},
  year = {1982},
  number = {1},
}

@misc{Kidger_sympy2jax_2025,
  author       = {Patrick Kidger},
  title        = {sympy2jax: Turn SymPy expressions into trainable JAX expressions},
  year         = {2025},
  howpublished = {\url{https://github.com/patrick-kidger/sympy2jax}},
  note         = {Version 0.0.7. Apache-2.0 license.}
}

@misc{jax2018github,
  author       = {Bradbury, James and Frostig, Roy and Hawkins, Peter and Johnson, Matthew James and
                  Leary, Chris and Maclaurin, Dougal and Necula, George and Paszke, Adam and
                  VanderPlas, Jake and Wanderman-Milne, Skye and Zhang, Qiao},
  title        = {{JAX}: Composable transformations of {P}ython+{N}um{P}y programs},
  year         = {2018},
  howpublished = {\url{https://github.com/google/jax}},
  note         = {Version 0.6.0.}
}

@article{sympy,
 title = {SymPy: symbolic computing in Python},
 author = {Meurer, Aaron and Smith, Christopher P. and Paprocki, Mateusz and \v{C}ert\'{i}k, Ond\v{r}ej and Kirpichev, Sergey B. and Rocklin, Matthew and Kumar, AMiT and Ivanov, Sergiu and Moore, Jason K. and Singh, Sartaj and Rathnayake, Thilina and Vig, Sean and Granger, Brian E. and Muller, Richard P. and Bonazzi, Francesco and Gupta, Harsh and Vats, Shivam and Johansson, Fredrik and Pedregosa, Fabian and Curry, Matthew J. and Terrel, Andy R. and Rou\v{c}ka, \v{S}t\v{e}p\'{a}n and Saboo, Ashutosh and Fernando, Isuru and Kulal, Sumith and Cimrman, Robert and Scopatz, Anthony},
 year = 2017,
 month = jan,
 volume = 3,
 pages = {e103},
 journal = {PeerJ Computer Science},
 issn = {2376-5992},
 url = {https://doi.org/10.7717/peerj-cs.103},
 doi = {10.7717/peerj-cs.103}
}

@ARTICLE{Lee2019-wf,
  title        = {{Linking Gaussian process regression with data-driven manifold
                  embeddings for nonlinear data fusion}},
  author       = {Lee, Seungjoon and Dietrich, Felix and Karniadakis, George E
                  and Kevrekidis, Ioannis G},
  journaltitle = {Interface Focus},
  publisher    = {The Royal Society},
  volume       = {9},
  issue        = {3},
  pages        = {20180083},
  date         = {2019-06-06},
  urldate      = {2025-11-09},
  language     = {en},
  journal = {Interface Focus},
  year = {2019},
  number = {3},
}

@ARTICLE{Pan2020-nd,
  title        = {{On the structure of time-delay embedding in linear models of
                  non-linear dynamical systems}},
  author       = {Pan, Shaowu and Duraisamy, Karthik},
  journaltitle = {Chaos},
  publisher    = {AIP Publishing},
  volume       = {30},
  issue        = {7},
  pages        = {073135},
  date         = {2020-07-28},
  urldate      = {2025-11-09},
  language     = {en},
  journal = {Chaos},
  year = {2020},
  number = {7},
}

@ARTICLE{Broomhead1989-bh,
  title        = {{Time-series analysis}},
  author       = {Broomhead, D S and Jones, R},
  journaltitle = {Proc. R. Soc. Lond.},
  publisher    = {The Royal Society},
  volume       = {423},
  issue        = {1864},
  pages        = {103--121},
  date         = {1989-05-08},
  urldate      = {2025-11-09},
  language     = {en},
  journal = {Proc. R. Soc. Lond.},
  year = {1989},
  number = {1864},
}

@ARTICLE{Bakarji2023-sq,
  title        = {{Discovering governing equations from partial measurements
                  with deep delay autoencoders}},
  author       = {Bakarji, Joseph and Champion, Kathleen and Nathan Kutz, J and
                  Brunton, Steven L},
  journaltitle = {Proc. Math. Phys. Eng. Sci.},
  publisher    = {The Royal Society},
  volume       = {479},
  issue        = {2276},
  date         = {2023-08-30},
  urldate      = {2025-11-09},
  language     = {en},
  journal = {Proc. Math. Phys. Eng. Sci.},
  year = {2023},
  number = {2276},
}

@article{donoho2006compressed,
  title        = {Compressed sensing},
  author       = {Donoho, David L.},
  journal      = {IEEE Transactions on Information Theory},
  volume       = {52},
  number       = {4},
  pages        = {1289--1306},
  year         = {2006},
  doi          = {10.1109/TIT.2006.871582}
}

@article{candes2006robust,
  title        = {Robust uncertainty principles: exact signal reconstruction from highly incomplete frequency information},
  author       = {Cand{\`e}s, Emmanuel J. and Romberg, Justin and Tao, Terence},
  journal      = {IEEE Transactions on Information Theory},
  volume       = {52},
  number       = {2},
  pages        = {489--509},
  year         = {2006},
  doi          = {10.1109/TIT.2005.862083}
}

@article{candes2008intro,
  title        = {An introduction to compressive sampling},
  author       = {Cand{\`e}s, Emmanuel J. and Wakin, Michael B.},
  journal      = {IEEE Signal Processing Magazine},
  volume       = {25},
  number       = {2},
  pages        = {21--30},
  year         = {2008},
  doi          = {10.1109/MSP.2007.914731}
}

@article{wang2011,
  title = {Predicting Catastrophes in Nonlinear Dynamical Systems by Compressive Sensing},
  author = {Wang, Wen-Xu and Yang, Rui and Lai, Ying-Cheng and Kovanis, Vassilios and Grebogi, Celso},
  journal = {Phys. Rev. Lett.},
  volume = {106},
  issue = {15},
  pages = {154101},
  numpages = {4},
  year = {2011},
  month = {4},
  publisher = {American Physical Society},
  doi = {10.1103/PhysRevLett.106.154101},
  url = {https://link.aps.org/doi/10.1103/PhysRevLett.106.154101},
  number = {15},
}

@ARTICLE{ZhouDerivative,
  title        = {{Derivative reproducing properties for kernel methods in
                  learning theory}},
  author       = {Zhou, Ding-Xuan},
  journaltitle = {J. Comput. Appl. Math.},
  publisher    = {Elsevier BV},
  volume       = {220},
  issue        = {1-2},
  pages        = {456--463},
  date         = {2008-10},
  language     = {en},
  journal = {J. Comput. Appl. Math.},
  year = {2008},
  number = {1-2},
}

@BOOK{Owhadi2019-xh,
  title     = {{Operator-Adapted Wavelets, Fast Solvers, and Numerical
               Homogenization}},
  author    = {Owhadi, Houman and Scovel, Clint},
  publisher = {Cambridge University Press},
  edition   = {1},
  date      = {2019-10-31},
  urldate   = {2025-10-02},
  year = {2019},
}

@book{Rasmussen2006Gaussian,
  author = {Rasmussen, Carl Edward and Williams, Christopher K. I.},
  publisher = {The MIT Press},
  title = {Gaussian Processes for Machine Learning},
  year = 2006
}

@misc{jalalian2025dataefficientkernelmethodslearning,
      title={Data-Efficient Kernel Methods for Learning Differential Equations and Their Solution Operators: Algorithms and Error Analysis}, 
      author={Yasamin Jalalian and Juan Felipe Osorio Ramirez and Alexander Hsu and Bamdad Hosseini and Houman Owhadi},
      year={2025},
      eprint={2503.01036},
      archivePrefix={arXiv},
      primaryClass={stat.ML},
      url={https://arxiv.org/abs/2503.01036}, 
}

@misc{kanagawa2018gaussianprocesseskernelmethods,
      title={Gaussian Processes and Kernel Methods: A Review on Connections and Equivalences}, 
      author={Motonobu Kanagawa and Philipp Hennig and Dino Sejdinovic and Bharath K Sriperumbudur},
      year={2018},
      eprint={1807.02582},
      archivePrefix={arXiv},
      primaryClass={stat.ML},
      url={https://arxiv.org/abs/1807.02582}, 
}

@BOOK{Wendland2010-gl,
  title     = {{Scattered data approximation}},
  author    = {Wendland, Holger},
  publisher = {Cambridge University Press},
  location  = {Cambridge, England},
  date      = {2010-02-11},
  pagetotal = {348},
  series    = {Cambridge monographs on applied and computational mathematics},
  urldate   = {2025-02-12},
  language  = {en},
  year = {2010},
}

@InProceedings{xu2025efficientkernelbasedsolversnonlinear,
  title = 	 {Toward Efficient Kernel-Based Solvers for Nonlinear {PDE}s},
  author =       {Xu, Zhitong and Long, Da and Xu, Yiming and Yang, Guang and Zhe, Shandian and Owhadi, Houman},
  booktitle = 	 {Proceedings of the 42nd International Conference on Machine Learning},
  pages = 	 {69526--69544},
  year = 	 {2025},
  editor = 	 {Singh, Aarti and Fazel, Maryam and Hsu, Daniel and Lacoste-Julien, Simon and Berkenkamp, Felix and Maharaj, Tegan and Wagstaff, Kiri and Zhu, Jerry},
  volume = 	 {267},
  series = 	 {Proceedings of Machine Learning Research},
  month = 	 {13--19 Jul},
  publisher =    {PMLR},
  url = 	 {https://proceedings.mlr.press/v267/xu25u.html},
}

@book{hairer2002solving-ii,
   address={Berlin},
   author={Hairer, E. and Wanner, G.},
   edition={Second Revised Edition},
   publisher={Springer},
   title={{S}olving {O}rdinary {D}ifferential {E}quations {II} {S}tiff and
          {D}ifferential-{A}lgebraic {P}roblems},
   year={2002}
 }

@article{soderlind2002automatic,
   title={Automatic control and adaptive time-stepping},
   author={Gustaf S{\"o}derlind},
   year={2002},
   journal={Numerical Algorithms},
   volume={31},
   pages={281--310}
}

@phdthesis{kidger2021on,
    title={{O}n {N}eural {D}ifferential {E}quations},
    author={Patrick Kidger},
    year={2021},
    school={University of Oxford},
}

@article{kidger2021equinox,
    author={Patrick Kidger and Cristian Garcia},
    title={{E}quinox: neural networks in {JAX} via callable {P}y{T}rees and
           filtered transformations},
    year={2021},
    journal={Differentiable Programming workshop at Neural Information Processing
             Systems 2021}
}

@article{tsitouras2011runge,
  title={Runge--Kutta pairs of order 5 (4) satisfying only the first column
         simplifying assumption},
  author={Tsitouras, Ch},
  journal={Computers \& Mathematics with Applications},
  volume={62},
  number={2},
  pages={770--775},
  year={2011},
  publisher={Elsevier}
}

@ARTICLE{Hokanson2023-sr,
  title        = {{Simultaneous identification and denoising of dynamical
                  systems}},
  author       = {Hokanson, Jeffrey M and Iaccarino, Gianluca and Doostan,
                  Alireza},
  journaltitle = {SIAM J. Sci. Comput.},
  publisher    = {Society for Industrial \& Applied Mathematics (SIAM)},
  volume       = {45},
  issue        = {4},
  pages        = {A1413--A1437},
  date         = {2023-08-31},
  language     = {en},
  journal = {SIAM J. Sci. Comput.},
  year = {2023},
  number = {4},
}

@ARTICLE{Ribera2022-nz,
    author = {Ribera, H. and Shirman, S. and Nguyen, A. V. and Mangan, N. M.},
    title = {Model selection of chaotic systems from data with hidden variables using sparse data assimilation},
    journal = {Chaos: An Interdisciplinary Journal of Nonlinear Science},
    volume = {32},
    number = {6},
    pages = {063101},
    year = {2022},
    month = {06},
    issn = {1054-1500},
    doi = {10.1063/5.0066066},
    url = {https://doi.org/10.1063/5.0066066},
}

@misc{Fung2025-odr,
      title={Overcoming error-in-variable problem in data-driven model discovery by orthogonal distance regression}, 
      author={Lloyd Fung},
      year={2025},
      eprint={2507.23426},
      archivePrefix={arXiv},
      primaryClass={stat.ME},
      url={https://arxiv.org/abs/2507.23426}, 
}

@ARTICLE{Aravkin2024-bz,
  title        = {{A Levenberg–Marquardt method for nonsmooth regularized least
                  squares}},
  author       = {Aravkin, Aleksandr Y and Baraldi, Robert and Orban, Dominique},
  journaltitle = {SIAM J. Sci. Comput.},
  publisher    = {Society for Industrial \& Applied Mathematics (SIAM)},
  volume       = {46},
  issue        = {4},
  pages        = {A2557--A2581},
  date         = {2024-08-31},
  urldate      = {2025-02-12},
  language     = {en},
  journal = {SIAM J. Sci. Comput.},
  year = {2024},
  number = {4},
}

@article{George2021Decomposing,
  author       = {Erin George and Colleen E. Chan and Gal Dimand and Ryan M. Chakmak and Claudia Falcon and Daniel Eckhardt and Robert Scott Martin},
  title        = {Decomposing Signals from Dynamical Systems Using Shadow Manifold Interpolation},
  journal      = {SIAM Journal on Applied Dynamical Systems},
  volume       = {20},
  number       = {4},
  pages        = {2236--2260},
  year         = {2021},
  doi          = {10.1137/20M1350923},
}

@article{Brunton2016,
    author = {Brunton, Steven L and Proctor, Joshua L and {Nathan Kutz}, J},
    doi = {10.1073/pnas.1517384113},
    journal = {Proceedings of the National Academy of Sciences},
    number = {15},
    title = {{Discovering governing equations from data by sparse identification of nonlinear dynamical systems}},
    volume = {113},
    year = {2016}
}

@article{Messenger2021a,
   author = {Daniel A. Messenger and David M. Bortz},
   doi = {10.1016/J.JCP.2021.110525},
   issn = {10902716},
   journal = {Journal of Computational Physics},
   month = {10},
   publisher = {Academic Press Inc.},
   title = {Weak SINDy for partial differential equations},
   volume = {443},
   year = {2021},
}

@article{Kaptanoglu2022, 
doi = {10.21105/joss.03994}, 
url = {https://doi.org/10.21105/joss.03994}, 
year = {2022}, 
publisher = {The Open Journal}, 
volume = {7}, 
number = {69}, 
pages = {3994}, 
author = {Kaptanoglu, Alan A. and de Silva, Brian M. and Fasel, Urban and Kaheman, Kadierdan and Goldschmidt, Andy J. and Callaham, Jared and Delahunt, Charles B. and Nicolaou, Zachary G. and Champion, Kathleen and Loiseau, Jean-Christophe and Kutz, J. Nathan and Brunton, Steven L.}, 
title = {PySINDy: A comprehensive Python package for robust sparse system identification}, journal = {Journal of Open Source Software} }

@article{haas2024kalman,
  author={Stevens-Haas, Jacob M. and Bhangale, Yash and Nathan Kutz, J. and Aravkin, Aleksandr},
  journal={IEEE Access}, 
  title={Learning Nonlinear Dynamics Using Kalman Smoothing}, 
  year={2024},
  volume={12},
  number={},
  pages={138564-138574},
  doi={10.1109/ACCESS.2024.3465390}}

@article{Champion2020,
   author = {Kathleen Champion and Peng Zheng and Aleksandr Y. Aravkin and Steven L. Brunton and J. Nathan Kutz},
   doi = {10.1109/ACCESS.2020.3023625},
   issn = {21693536},
   journal = {IEEE Access},
   pages = {169259-169271},
   title = {A unified sparse optimization framework to learn parsimonious physics-informed models from data},
   volume = {8},
   year = {2020},
}

@article{fasel_ensemble-sindy_2022,
	title = {Ensemble-{SINDy}: Robust sparse model discovery in the low-data, high-noise limit, with active learning and control},
	volume = {478},
	url = {https://royalsocietypublishing.org/doi/full/10.1098/rspa.2021.0904},
	doi = {10.1098/rspa.2021.0904},
	shorttitle = {Ensemble-{SINDy}},
	pages = {20210904},
	number = {2260},
	journaltitle = {Proceedings of the Royal Society A: Mathematical, Physical and Engineering Sciences},
	author = {Fasel, U. and Kutz, J. N. and Brunton, B. W. and Brunton, S. L.},
	urldate = {2024-11-12},
	date = {2022-04-13},
	note = {Publisher: Royal Society},
  journal = {Proceedings of the Royal Society A: Mathematical, Physical and Engineering Sciences},
  year = {2022},
}

@article{Kaheman_2022,
doi = {10.1088/2632-2153/ac567a},
url = {https://doi.org/10.1088/2632-2153/ac567a},
year = {2022},
month = {3},
publisher = {IOP Publishing},
volume = {3},
number = {1},
pages = {015031},
author = {Kaheman, Kadierdan and Brunton, Steven L and Nathan Kutz, J},
title = {Automatic differentiation to simultaneously identify nonlinear dynamics and extract noise probability distributions from data},
journal = {Machine Learning: Science and Technology},
}

@article{brunton_chaos_2017,
	title = {Chaos as an intermittently forced linear system},
	volume = {8},
	rights = {2017 The Author(s)},
	issn = {2041-1723},
	url = {https://www.nature.com/articles/s41467-017-00030-8},
	doi = {10.1038/s41467-017-00030-8},
	pages = {19},
	number = {1},
	journaltitle = {Nature Communications},
	shortjournal = {Nat Commun},
	author = {Brunton, Steven L. and Brunton, Bingni W. and Proctor, Joshua L. and Kaiser, Eurika and Kutz, J. Nathan},
	urldate = {2025-09-30},
	date = {2017-05-30},
	langid = {english},
	note = {Publisher: Nature Publishing Group},
  journal = {Nature Communications},
  year = {2017},
}

@inproceedings{rand_detecting_1981,
	location = {Berlin, Heidelberg},
	title = {Detecting strange attractors in turbulence},
	volume = {898},
	rights = {http://www.springer.com/tdm},
	isbn = {978-3-540-38945-3},
	url = {http://link.springer.com/10.1007/BFb0091924},
	pages = {366--381},
	booktitle = {Dynamical Systems and Turbulence, Warwick 1980},
	publisher = {Springer Berlin Heidelberg},
	author = {Takens, Floris},
	editor = {Rand, David and Young, Lai-Sang},
	urldate = {2025-09-30},
	date = {1981},
	langid = {english},
	doi = {10.1007/BFb0091924},
  year = {1981},
}

@article{Gao2022,
    author = {Mars Gao, L. and Nathan Kutz, J.},
    title = {Bayesian autoencoders for data-driven discovery of coordinates, governing equations and fundamental constants},
    journal = {Proceedings of the Royal Society A: Mathematical, Physical and Engineering Sciences},
    volume = {480},
    number = {2286},
    pages = {20230506},
    year = {2024},
    month = {03},
    issn = {1364-5021},
    doi = {10.1098/rspa.2023.0506},
    url = {https://doi.org/10.1098/rspa.2023.0506},
    eprint = {https://royalsocietypublishing.org/rspa/article-pdf/doi/10.1098/rspa.2023.0506/1228812/rspa.2023.0506.pdf},
}

@article{bertsimas_learning_2023,
	title = {Learning sparse nonlinear dynamics via mixed-integer optimization},
	volume = {111},
	issn = {1573-269X},
	url = {https://doi.org/10.1007/s11071-022-08178-9},
	doi = {10.1007/s11071-022-08178-9},
	pages = {6585--6604},
	number = {7},
	journaltitle = {Nonlinear Dynamics},
	shortjournal = {Nonlinear Dyn},
	author = {Bertsimas, Dimitris and Gurnee, Wes},
	urldate = {2024-11-25},
	date = {2023-04-01},
	langid = {english},
  journal = {Nonlinear Dynamics},
  year = {2023},
}

@article{ribeiro2020smoothness,
  author  = {Ribeiro, Ant\^onio H. and Tiels, Koen and Umenberger, Jack and Sch\"on, Thomas B. and Aguirre, Luis A.},
  title   = {On the smoothness of nonlinear system identification},
  journal = {Automatica},
  year    = {2020},
  volume  = {121},
  pages   = {109158},
  doi     = {10.1016/j.automatica.2020.109158},
}

@article{Carbonell2015multishooting,
title = {Multiple Shooting-Local Linearization method for the identification of dynamical systems},
journal = {Communications in Nonlinear Science and Numerical Simulation},
volume = {37},
pages = {292-304},
year = {2016},
issn = {1007-5704},
doi = {https://doi.org/10.1016/j.cnsns.2016.01.022},
url = {https://www.sciencedirect.com/science/article/pii/S1007570416300107},
author = {F. Carbonell and Y. Iturria-Medina and J.C. Jimenez},
}

@article{Chakraborty2024multistep,
title = {Divide and conquer: Learning chaotic dynamical systems with multistep penalty neural ordinary differential equations},
journal = {Computer Methods in Applied Mechanics and Engineering},
volume = {432},
pages = {117442},
year = {2024},
issn = {0045-7825},
doi = {https://doi.org/10.1016/j.cma.2024.117442},
url = {https://www.sciencedirect.com/science/article/pii/S0045782524006972},
author = {Dibyajyoti Chakraborty and Seung Whan Chung and Troy Arcomano and Romit Maulik},
}

@ARTICLE{Chung2022-hz,
  title        = {{An optimization method for chaotic turbulent flow}},
  author       = {Chung, Seung Whan and Freund, Jonathan B},
  journaltitle = {J. Comput. Phys.},
  publisher    = {Elsevier BV},
  volume       = {457},
  issue        = {111077},
  pages        = {111077},
  date         = {2022-05},
  language     = {en},
  journal = {J. Comput. Phys.},
  year = {2022},
  number = {111077},
}

@techreport{householder1950prony,
  title={On Prony's method of fitting exponential decay curves and multiple-hit survival curves},
  author={Householder, Alston Scott},
  institution = {Oak Ridge National Laboratory},
  year={1950}
}

@article{Bell1994IteratedKalmanSmoother,
  author  = {Bell, Bradley M.},
  title   = {The Iterated {Kalman} Smoother as a {Gauss--Newton} Method},
  journal = {SIAM Journal on Optimization},
  volume  = {4},
  number  = {3},
  pages   = {626--636},
  year    = {1994},
  doi     = {10.1137/0804035}
}

@book{Sarkka2013BayesianFilteringSmoothing,
  author    = {S{\"a}rkk{\"a}, Simo},
  title     = {Bayesian Filtering and Smoothing},
  publisher = {Cambridge University Press},
  year      = {2013},
  isbn      = {9781107619289}
}

@article{tibshirani1996lasso,
 ISSN = {00359246},
 URL = {http://www.jstor.org/stable/2346178},
 author = {Robert Tibshirani},
 journal = {Journal of the Royal Statistical Society. Series B (Methodological)},
 number = {1},
 pages = {267--288},
 publisher = {[Royal Statistical Society, Oxford University Press]},
 title = {Regression Shrinkage and Selection via the Lasso},
 urldate = {2025-10-31},
 volume = {58},
 year = {1996}
}

@article{Chen2021,
  author = {Zhao Chen and Yang Liu and Hao Sun},
  title = {Physics-informed learning of governing equations from scarce data},
  journal = {Nature Communications},
  volume = {12},
  article = {6136},
  year = {2021},
  doi = {10.1038/s41467-021-26434-1},
  url = {https://www.nature.com/articles/s41467-021-26434-1},
  publisher = {Nature Publishing Group}
}

@inproceedings{ijcai2021p283,
    author = {Sun, Fangzheng and Liu, Yang and Sun, Hao},
    booktitle = {Proceedings of the Thirtieth International Joint Conference on Artificial Intelligence, {IJCAI-21}},
    doi = {10.24963/ijcai.2021/283},
    editor = {Zhi-Hua Zhou},
    month = {8},
    note = {Main Track},
    pages = {2054--2061},
    publisher = {International Joint Conferences on Artificial Intelligence Organization},
    title = {Physics-informed Spline Learning for Nonlinear Dynamics Discovery},
    url = {https://doi.org/10.24963/ijcai.2021/283},
    year = {2021}}

@article{Lahouel2024,
  author = {K. Lahouel},
  title = {Learning nonparametric ordinary differential equations from noisy data},
  journal = {Journal of Computational Physics},
  volume = {479},
  pages = {111019},
  year = {2024},
  doi = {10.1016/j.jcp.2023.111019},
  url = {https://doi.org/10.1016/j.jcp.2023.111019},
  publisher = {Elsevier}
}

@article{rudy_soft,
    author = {Samuel H. Rudy and Steven L. Brunton and J. Nathan Kutz},
    doi = {https://doi.org/10.1016/j.jcp.2019.108860},
    issn = {0021-9991},
    journal = {Journal of Computational Physics},
    pages = {108860},
    title = {Smoothing and parameter estimation by soft-adherence to governing equations},
    url = {https://www.sciencedirect.com/science/article/pii/S0021999119305443},
    volume = {398},
    year = {2019}}

@article{CHEN2021110668,
    author = {Yifan Chen and Bamdad Hosseini and Houman Owhadi and Andrew M. Stuart},
    doi = {https://doi.org/10.1016/j.jcp.2021.110668},
    issn = {0021-9991},
    journal = {Journal of Computational Physics},
    pages = {110668},
    title = {Solving and learning nonlinear PDEs with Gaussian processes},
    url = {https://www.sciencedirect.com/science/article/pii/S0021999121005635},
    volume = {447},
    year = {2021},
}

@article{LONG2024134095,
	author = {Da Long and Nicole Mrvaljevi{\'c} and Shandian Zhe and Bamdad Hosseini},
	doi = {https://doi.org/10.1016/j.physd.2024.134095},
	issn = {0167-2789},
	journal = {Physica D: Nonlinear Phenomena},
	pages = {134095},
	title = {A kernel framework for learning differential equations and their solution operators},
	url = {https://www.sciencedirect.com/science/article/pii/S0167278924000460},
	volume = {460},
	year = {2024}}

@article{tristan_pde,
  title={A penalty method for PDE-constrained optimization in inverse problems},
  author={Leeuwen, T van and Herrmann, Felix J},
  journal={Inverse Problems},
  volume={32},
  number={1},
  pages={015007},
  year={2016},
  publisher={IOP Publishing}
}

@article{sasha_tristan_partialmin,
  author  = {Aravkin, Aleksandr Y. and Drusvyatskiy, Dmitriy and van Leeuwen, Tristan},
  title   = {Efficient Quadratic Penalization Through the Partial Minimization Technique},
  journal = {IEEE Transactions on Automatic Control},
  year    = {2018},
  volume  = {63},
  number  = {7},
  pages   = {2131--2138},
  doi     = {10.1109/TAC.2017.2754474},
  url     = {https://doi.org/10.1109/TAC.2017.2754474}
}

@ARTICLE{tristan_analysis,
  title        = {{An analysis of constraint-relaxation in PDE-based inverse
                  problems}},
  author       = {van Leeuwen, Tristan and Yang, Yunan},
  journaltitle = {Inverse Problems},
  volume       = {41},
  issue        = {2},
  pages        = {025009},
  date         = {2025-02-28},
  journal = {Inverse Problems},
  year = {2025},
  number = {2},
}

@misc{rackauckas2020universal,
  title        = {{Universal differential equations for scientific machine
                  learning}},
  author       = {Rackauckas, Christopher and Ma, Yingbo and Martensen, Julius
                  and Warner, Collin and Zubov, Kirill and Supekar, Rohit and
                  Skinner, Dominic and Ramadhan, Ali and Edelman, Alan},
  date         = {2020-08-31},
  eprinttype   = {arXiv},
  eprintclass  = {cs.LG},
  year = {2020},
}

@inproceedings{sun2022bayesiansplinelearningequation,
  author = {Sun, Luning and Huang, Daniel Zhengyu and Sun, Hao and Wang, Jian-Xun},
  title = {Bayesian spline learning for equation discovery of nonlinear dynamics with quantified uncertainty},
  year = {2022},
  isbn = {9781713871088},
  publisher = {Curran Associates Inc.},
  address = {Red Hook, NY, USA},
  booktitle = {Proceedings of the 36th International Conference on Neural Information Processing Systems},
  articleno = {502},
  numpages = {14},
  location = {New Orleans, LA, USA},
  series = {NIPS '22}
}

@misc{ensemble_uncertainty,
      title={Convergence of uncertainty estimates in Ensemble and Bayesian sparse model discovery}, 
      author={L. Mars Gao and Urban Fasel and Steven L. Brunton and J. Nathan Kutz},
      year={2023},
      eprint={2301.12649},
      archivePrefix={arXiv},
      primaryClass={cs.LG},
      url={https://arxiv.org/abs/2301.12649}, 
}

@book{ljungTheory,
author = {Ljung, Lennart},
title = {System identification (2nd ed.): theory for the user},
year = {1999},
isbn = {0136566952},
publisher = {Prentice Hall PTR},
address = {USA}
}

@book{Jazwinski1970,
  author    = {Andrew H. Jazwinski},
  title     = {Stochastic Processes and Filtering Theory},
  publisher = {Academic Press},
  year      = {1970},
  address   = {New York},
}

@article{vanBreugel2020NumericalDifferentiation,
    doi={10.1109/ACCESS.2020.3034077},
    author={F. {van Breugel} and J. {Nathan Kutz} and B. W. {Brunton}},
    journal={IEEE Access},
    title={Numerical differentiation of noisy data: A unifying multi-objective optimization framework},
    year={2020}
    }

@article{gevers_personal_2006,
	title = {A personal view of the development of system identification: A 30-year journey through an exciting field},
	volume = {26},
	issn = {1941-000X},
	url = {https://ieeexplore.ieee.org/document/4019326/},
	doi = {10.1109/MCS.2006.252834},
	shorttitle = {A personal view of the development of system identification},
	pages = {93--105},
	number = {6},
	journaltitle = {{IEEE} Control Systems Magazine},
	author = {{GEVERS}, {MICHEL}},
	urldate = {2025-10-31},
	date = {2006-12},
  journal = {{IEEE} Control Systems Magazine},
  year = {2006},
}

@article{vanLeeuwen2013Mitigating,
  author   = {Tristan van Leeuwen and Felix J. Herrmann},
  title    = {Mitigating Local Minima in Full-Waveform Inversion by Expanding the Search Space},
  journal  = {Geophysical Journal International},
  volume   = {195},
  number   = {1},
  pages    = {661--667},
  year     = {2013},
  doi      = {10.1093/gji/ggt258}
}

@ARTICLE{Chartrand2011-cv,
  title        = {{Numerical differentiation of noisy, nonsmooth data}},
  author       = {Chartrand, Rick},
  journaltitle = {ISRN Appl. Math.},
  publisher    = {Hindawi Limited},
  volume       = {2011},
  issue        = {1},
  pages        = {1--11},
  date         = {2011-05-11},
  urldate      = {2025-11-04},
  language     = {en},
  journal = {ISRN Appl. Math.},
  year = {2011},
  number = {1},
}

@article{Raissi2018Hidden,
  author  = {Maziar Raissi},
  title   = {Deep hidden physics models: Deep learning of nonlinear partial differential equations},
  journal = {Journal of Computational Physics},
  volume  = {357},
  pages   = {125--141},
  year    = {2018},
  doi     = {10.1016/j.jcp.2017.11.038}
}

@ARTICLE{Raissi2019-so,
  title        = {{Physics-informed neural networks: A deep learning framework
                  for solving forward and inverse problems involving nonlinear
                  partial differential equations}},
  author       = {Raissi, M and Perdikaris, P and Karniadakis, G E},
  journaltitle = {J. Comput. Phys.},
  volume       = {378},
  pages        = {686--707},
  date         = {2019-02-01},
  journal = {J. Comput. Phys.},
  year = {2019},
}

@ARTICLE{bayesiansindy,
  title        = {{Sparsifying priors for Bayesian uncertainty quantification in
                  model discovery}},
  author       = {Hirsh, Seth M and Barajas-Solano, David A and Kutz, J Nathan},
  journaltitle = {R. Soc. Open Sci.},
  publisher    = {The Royal Society},
  volume       = {9},
  issue        = {2},
  pages        = {211823},
  date         = {2022-02-23},
  urldate      = {2025-10-11},
  language     = {en},
  journal = {R. Soc. Open Sci.},
  year = {2022},
  number = {2},
}

@inproceedings{Chen2018,
  author    = {Chen, Tian Qi and Rubanova, Yulia and Bettencourt, Jesse and Duvenaud, David},
  title     = {Neural Ordinary Differential Equations},
  booktitle = {Advances in Neural Information Processing Systems (NeurIPS 31)},
  pages     = {6572--6583},
  year      = {2018}
}

@ARTICLE{Givoli2021-bc,
  title        = {{A tutorial on the adjoint method for inverse problems}},
  author       = {Givoli, Dan},
  journaltitle = {Comput. Methods Appl. Mech. Eng.},
  publisher    = {Elsevier BV},
  volume       = {380},
  issue        = {113810},
  pages        = {113810},
  date         = {2021-07-01},
  urldate      = {2025-10-11},
  language     = {en},
  journal = {Comput. Methods Appl. Mech. Eng.},
  year = {2021},
  number = {113810},
}

@book{Lions1971,
  author    = {Lions, Jacques-Louis},
  title     = {Optimal Control of Systems Governed by Partial Differential Equations},
  publisher = {Springer-Verlag},
  year      = {1971}
}

@ARTICLE{unifying-representer,
  title        = {{A unifying view of representer theorems}},
  author       = {Argyriou, Andreas and Dinuzzo, Francesco},
  editor       = {Xing, Eric P and Jebara, Tony},
  journaltitle = {ICML},
  publisher    = {PMLR},
  volume       = {32},
  issue        = {2},
  pages        = {748--756},
  date         = {2014-06-21},
  journal = {ICML},
  year = {2014},
  number = {2},
}

@InProceedings{Scholkopf2001-co,
author="Sch{\"o}lkopf, Bernhard
and Herbrich, Ralf
and Smola, Alex J.",
editor="Helmbold, David
and Williamson, Bob",
title="A Generalized Representer Theorem",
booktitle="Computational Learning Theory",
year="2001",
publisher="Springer Berlin Heidelberg",
address="Berlin, Heidelberg",
pages="416--426",
isbn="978-3-540-44581-4"
}

@book{evansPDE,
  author    = {Evans, Lawrence C.},
  title     = {Partial Differential Equations},
  edition   = {2nd},
  series    = {Graduate Studies in Mathematics},
  volume    = {19},
  publisher = {American Mathematical Society},
  address   = {Providence, RI},
  year      = {2010},
  isbn      = {978-0-8218-4974-3},
}

@techreport{smith_application_1962,
	title = {Application of Statistical Filter Theory to the Optimal Estimation of Position and Velocity on Board a Circumlunar Vehicle},
	url = {https://ntrs.nasa.gov/citations/20190002215},
	number = {{NASA}-{TR}-R-135},
	author = {Smith, Gerald L. and Schmidt, Stanley F. and {McGee}, Leonard A.},
	urldate = {2025-10-31},
	date = {1962-01-01},
	note = {{NTRS} Author Affiliations: {NASA} Ames Research Center
{NTRS} Document {ID}: 20190002215
{NTRS} Research Center: Ames Research Center ({ARC})},
  year = {1962},
  institution = {Ames Research Center}
}

@article{aastrom1971system,
  title={System identification—a survey},
  author={{\AA}str{\"o}m, Karl Johan and Eykhoff, Peter},
  journal={Automatica},
  volume={7},
  number={2},
  pages={123--162},
  year={1971},
  publisher={Elsevier}
}

@article{ljung2010perspectives,
  title={Perspectives on system identification},
  author={Ljung, Lennart},
  journal={Annual Reviews in Control},
  volume={34},
  number={1},
  pages={1--12},
  year={2010},
  publisher={Elsevier}
}

@book{keesman2011system,
  title={System identification: an introduction},
  author={Keesman, Karel J},
  year={2011},
  publisher={Springer Science \& Business Media}
}

@article{bongard2007automated,
author = {Bongard, Josh and Lipson, Hod},
year = {2007},
month = {07},
pages = {9943-8},
title = {Automated reverse engineering of nonlinear dynamical systems},
volume = {104},
journal = {Proceedings of the National Academy of Sciences of the United States of America}
}

@article{schmidt2009distilling,
author = {Schmidt, Mariek and Lipson, Hod},
year = {2009},
month = {05},
pages = {81-5},
title = {Distilling Free-Form Natural Laws from Experimental Data},
volume = {324},
journal = {Science (New York, N.Y.)}
}

@article{schaeffer2017learning,
  title={Learning partial differential equations via data discovery and sparse optimization},
  author={Schaeffer, Hayden},
  journal={Proceedings of the Royal Society A: Mathematical, Physical and Engineering Sciences},
  volume={473},
  number={2197},
  pages={20160446},
  year={2017},
  publisher={The Royal Society Publishing}
}

@article{rudy2016datadriven,
author = {Samuel H. Rudy  and Steven L. Brunton  and Joshua L. Proctor  and J. Nathan Kutz },
title = {Data-driven discovery of partial differential equations},
journal = {Science Advances},
volume = {3},
number = {4},
pages = {e1602614},
year = {2017},
doi = {10.1126/sciadv.1602614},
URL = {https://www.science.org/doi/abs/10.1126/sciadv.1602614},
eprint = {https://www.science.org/doi/pdf/10.1126/sciadv.1602614},
}

@article{zhang2018robust,
  title={Robust data-driven discovery of governing physical laws with error bars},
  author={Zhang, Sheng and Lin, Guang},
  journal={Proceedings of the Royal Society A: Mathematical, Physical and Engineering Sciences},
  volume={474},
  number={2217},
  pages={20180305},
  year={2018},
  publisher={The Royal Society Publishing}
}

@inproceedings{niven2019bayesian,
  title={Bayesian identification of dynamical systems},
  author={Niven, Robert K and Mohammad-Djafari, Ali and Cordier, Laurent and Abel, Markus W and Quade, Markus},
  booktitle={39th International Workshop on Bayesian Inference and Maximum Entropy Methods in Science and Engineering (MaxEnt 2019)},
  volume={33},
  year={2019}
}

@article{yang2020bayesian,
  title={Bayesian differential programming for robust systems identification under uncertainty},
  author={Yang, Yibo and Aziz Bhouri, Mohamed and Perdikaris, Paris},
  journal={Proceedings of the Royal Society A},
  volume={476},
  number={2243},
  pages={20200290},
  year={2020},
  publisher={The Royal Society Publishing}
}

@article{north2022bayesian,
  title={A Bayesian approach for data-driven dynamic equation discovery},
  author={North, Joshua S and Wikle, Christopher K and Schliep, Erin M},
  journal={Journal of Agricultural, Biological and Environmental Statistics},
  volume={27},
  number={4},
  pages={728--747},
  year={2022},
  publisher={Springer}
}

@article{north2023bayesian,
  title={A Bayesian Approach for Spatio-Temporal Data-Driven Dynamic Equation Discovery},
  author={North, Joshua S and Wikle, Christopher K and Schliep, Erin M},
  journal={Bayesian Analysis},
  volume={1},
  number={1},
  pages={1--30},
  year={2023},
  publisher={International Society for Bayesian Analysis}
}
\appendix
\section{Theory}
We provide some more details of the RKHS theory that underlies our approach and the least squares relaxation approach we take for the ODE. 
\subsection{RKHS and representer theorems}\label{subsec:representer}

We first give a version of the representer theorem in the Hilbert space setting which is appropriate for \Cref{thm:obj-representer}. These theorems are well known and appear in various forms throughout the literature; we state this for completeness. The variant that we present can be found in \cite{unifying-representer,Owhadi2019-xh,Scholkopf2001-co}, and similar results appear in \cite{CHEN2021110668,ZhouDerivative,kanagawa2018gaussianprocesseskernelmethods}. 
\begin{theorem}[Representer]\label{thm:rep}
    Let \(\mathcal{H}\) be a separable Hilbert space. Let \(F:\mathcal{H} \mapsto \mathbb{R}\) be given by
    \begin{equation}
        F(u) = f(
        \langle u,w_1\rangle,...,\langle u,w_m\rangle) + J(\|u\|)
    \end{equation}
    where \(J:\mathbb{R}^+\mapsto\mathbb{R} \cup \{\infty\}\) is strictly increasing, \(\langle \cdot, w_i\rangle \) are bounded linear functionals, and \(f:\mathbb{R}^m \mapsto \mathbb{R}\cup \{\infty\}\). Then if \(F\) admits a minimizer \(u^*\) with \(F(u^*) < \infty\), then 
    \begin{equation}
        u^* = \sum_{i=1}^m c_i w_i. 
    \end{equation}
\end{theorem}

\subsection{Mat\'ern kernels and Sobolev spaces}\label{subsec:sobolev}
Here, we review some prerequisite RKHS and Sobolev space theory, giving more details on the kernel functions we use, and discussing the smoothness of functions in RKHS \(\X\) associated to the kernels we define in \cref{eqn:kernel}. Much of this will pull directly from \cite{kanagawa2018gaussianprocesseskernelmethods, Wendland2010-gl} where we direct the interested reader for additional details. First, recall from \cref{eqn:kernel} that we take the sum of a constant kernel with a Matérn kernel
\begin{equation}
    k(t,t') = c_0 + c_1 h_{\nu}\left(\frac{|t - t'|}{\ell}\right),\quad 
    h_{\nu}(r) = \frac{2^{1-\nu}}{\Gamma(\nu)}
\left(\sqrt{2\nu}r\right)^{\nu}
K_{\nu}\left(\sqrt{2\nu}\,r\right). 
\end{equation}
From \cite[Ex.~2.6]{kanagawa2018gaussianprocesseskernelmethods}, the RKHS \(\mathcal{H}_k\) associated to this kernel is norm equivalent to the Sobolev space \(H^{\nu +\frac{1}{2}}([0,T])\), functions which are \(\nu +\frac{1}{2}\) times differentiable in the mean-square sense. Compactness of the domain means that adding a constant to the kernel function does not affect the equivalence. By the Sobolev embedding thoerem \cite[Thm.~5.6.6, p.~270]{evansPDE}, this implies that \(\mathcal{H}_k\) is also compactly embedded in the H\"older space \(C^p([0,T])\) for \(p <\nu\). 

When \(\nu\) is a half integer (i.e. \(\nu = p+\frac{1}{2}\)), \(h_\nu\) reduces to the product of a polynomial and exponential in \(r\). 
\begin{equation}
    h_{\nu}(r) = \exp\left(-\sqrt{2\nu}r\right)\frac{p!}{(2p)!}\sum_{i=0}^p\frac{(\nu-\frac{1}{2}+i)!}{i!(\nu-\frac{1}{2}-i)!}\left(2\sqrt{2\nu}r\right)^{\nu-\frac{1}{2}-i}. 
\end{equation}
Stably computing high order derivatives of this \(h_\nu\) can still be challenging due to the absolute value inside of \(h_\nu\), and directly applying autodiff to this representation can cause catastrophic cancellation. However, for \(j < 2\nu\), \(\frac{d^j}{dt^j} h_{\nu}(\frac{|t - t'|}{\ell})\), 
can also be written as the product of a polynomial and exponential function; we compute this representation of the derivatives in SymPy \cite{sympy}. We then convert the resulting expressions to Jax \cite{jax2018github} using Sympy2Jax \cite{Kidger_sympy2jax_2025}.

\subsection{Cholesky factorization and choice of basis}\label{subsec:cholesky}
In our numerical implementation, we perform a linear change of variables in order to simplify bookkeeping and 
mitigate some of the effects of ill-conditioning and numerical roundoff on convergence by
reparametrizing with the Cholesky factors of \(\vect{K}\). That is,
write 
\begin{equation}
    \vect{K} = \vect{C}\vect{C}^\top
\end{equation}
for the Cholesky factorization of \(\vect{K}\), and set 
\begin{equation}
    \tilde{\vect{z}} = (\tilde{\vect{z}_1},...,\tilde{\vect{z}}_d),\quad \tilde{\vect{z}}_m = \vect{C}^\top \vect{z}_m.
\end{equation}
Under this parametrization, we have that 
\begin{equation}
    \vect{x}(\cdot;(\vect{C}^{-T}\tilde{\vect{z}}_1,...,\vect{C}^{-T}\tilde{\vect{z}}_d)) = \vect{x}(\cdot;\vect{z}). 
\end{equation}

It is clear that these are interchangeable \emph{mathematically}. 
Computationally, they relegate some of the bookkeeping with tracking the kernel matrix to an initial Cholesky factorization, which, in our implementation, is absorbed into the parametrization. Even though this still entails factorizations of ill-conditioned matrices, and direct solvers are also applied downstream, this empirically seems to mitigate some of the effects of ill-conditioned Jacobian matrices in solving regularized least-squares problems.

\section{Computational tools}
We perform all of our computations using Jax \cite{jax2018github}, and 
solve ODEs using the Tsit5 \cite{tsitouras2011runge} solver with a PID controller for adaptive timestepping \cite{hairer2002solving-ii,soderlind2002automatic} as recommended by, and implemented in Diffrax \cite{kidger2021equinox,kidger2021on}. Comparisons with SINDy use the PySINDy package \cite{Kaptanoglu2022}. Comparisons with ODR-BINDy were made using the implementation \footnote{\url{https://github.com/llfung/ODR-BINDy}} from the original work \cite{Fung2025-odr}. Our experiments and current implementation are available in the GitHub repository \cite{jsindyRepo}; an optimized implementation of JSINDy will soon be added to PySINDy. 
Our computations with JAX used an RTX-5000 Ada Generation GPU in a Lambda machine.

\end{document}